\documentclass[journal, twocolumn,letter,10pt]{IEEEtran}

\usepackage{multirow}
\usepackage{amsfonts,amsmath,amssymb,amsthm}
\makeatletter
\renewcommand*\env@matrix[1][*\c@MaxMatrixCols c]{%
  \hskip -\arraycolsep
  \let\@ifnextchar\new@ifnextchar
  \array{#1}}
\makeatother
\usepackage{authblk}
\usepackage{wasysym}
\usepackage{graphicx}
\usepackage{subcaption}
\usepackage{url}
\usepackage{color}
\usepackage{cite}

\newcommand{\be}{\begin{eqnarray}}
\newcommand{\ee}{\end{eqnarray}}

\newcommand{\ben}{\begin{enumerate}}
\newcommand{\een}{\end{enumerate}}
\newcommand{\beq}{\begin{equation}}
\newcommand{\eeq}{\end{equation}}
\newcommand{\beqa}{\begin{eqnarray*}}
\newcommand{\eeqa}{\end{eqnarray*}}
\newcommand{\bit}{\begin{itemize}}
\newcommand{\eit}{\end{itemize}}
\newcommand{\bt}{\begin{tabular}{c}}
\newcommand{\btt}{\begin{tabular}}
\newcommand{\et}{\end{tabular}}

\newtheorem{definition}{Definition}

\newtheorem{theorem}{Theorem}

\newtheorem{remark}{Remark}
\newtheorem{note}{Note}

\newcommand{\squishlist}{
   \begin{list}{$\bullet$}
    { \setlength{\itemsep}{0pt}      \setlength{\parsep}{0pt}
      \setlength{\topsep}{3pt}       \setlength{\partopsep}{0pt}
      \setlength{\listparindent}{-2pt}
      \setlength{\itemindent}{-5pt}
      \setlength{\leftmargin}{1em} \setlength{\labelwidth}{0em}
      \setlength{\labelsep}{0.5em} } }
\newcommand{\squishend}{
    \end{list}  }

\begin{document}
\title{An agile and distributed mechanism for inter-domain network slicing in next generation mobile networks}
\author[1, 2]{ Jalal Khamse-Ashari \thanks{* The paper has been accepted for publication at IEEE Transactions on Mobile Computing. Please refer to DOI 10.1109/TMC.2021.3061613 for the final version.}}
\author[2]{ Gamini Senarath}
\author[1, 2]{ Irem Bor-Yaliniz}
\author[1]{ Halim Yanikomeroglu}
\affil[1]{ Department of Systems and Computer Engineering, Carleton University, Ottawa, Canada}
\affil[2]{ Huawei Canada Research Center, Ottawa, Canada}
%Emails: \IEEEauthorrefmark{1}\{jalalkhamseashari,~irembor,~halim\}@sce.carleton.ca,~\IEEEauthorrefmark{2}\{gamini.senarath\}@huawei.com
\maketitle

\begin{abstract}
  Network slicing is emerging as a promising method to provide sought after versatility and flexibility to cope with ever-increasing demands. To realize such potential advantages and to meet the challenging requirements of various network slices in an on-demand fashion, we need to develop an agile and distributed mechanism for resource provisioning to different network slices in a heterogeneous multi-resource multi-domain mobile network environment.
   We formulate inter-domain resource provisioning to network slices in such an environment as an optimization problem
   which maximizes social welfare among network slice tenants (so that maximizing tenants' satisfaction),
   while minimizing operational expenditures for infrastructure service providers at the same time.
   To solve the envisioned problem, we implement an iterative auction game among network slice tenants,
  on one hand, and a plurality of price-taking subnet service providers, on the other hand.
  We show that the proposed solution method results in a distributed privacy-saving mechanism which converges to the \emph{optimal solution} of the described optimization problem.
  In addition to providing analytical results to characterize the performance of the proposed mechanism,
  we also employ numerical evaluations to validate the results, demonstrate convergence of the presented algorithm,
  and show the enhanced performance of the proposed approach (in terms of resource utilization, fairness and operational costs)
  against the existing solutions.
\end{abstract}


\section{Introduction}\label{sec:introduction}
Network slicing is expected to become an integral part of next generation mobile networks,
since slicing helps networks to be more versatile.
Built upon recently developed technologies, such as network function virtualization (NFV) and software defined networks (SDN),
network slicing lets multiple logical networks share the same physical infrastructure.
This virtualization technique can enable one network providing multiple services with extremely
different requirements, while maintaining isolation.
As opposed to many QoS assurance methods, network slicing differentiates between different traffic as per their requirements, and for the
same kind of traffic as per different tenants\footnote{Tenant can be thought as a group of users or a 3rd-party consumer using
the communication services to provide other communication services to users. }. Therefore, network slicing is a
unique method for end-to-end granular network management and service provisioning.
Moreover, by integrating with other technologies (such as cloud computing, cloud-RAN~\cite{checko2015cloud} and
mobile edge computing~\cite{taleb2017multi}), it introduces
more flexibility in deployment resulting in significant reduction in both operational and capital expenditures.
Accordingly, network slicing has attracted significant attention from both industry and academia.
While the leading standardization bodies for the next generation wireless networks included network slicing in their
work items~\cite{FCC_netSlicing, GSMA_NS, NFMN_NS, 3Gpp_NS}, there is also a significant number of studies from the academia investigating
characteristics and dynamics of network slicing \cite{foukas2017network, zhang2017network, zhou2016network, liang2015wireless, nikaein2015network,
costa2013radio, 5GmobileArch, samdanis2016network}.

From the control plane (CP) perspective in 3GPP 5G standardization,
a network slice can consist of several network slice instances
(NSIs). NSIs help not only to adjust the network slice resources
based on the service requirements, but also to differentiate between even the same type of QoS flows based on granular network policies.
The network management (NM) perspective involves also the network slice subnet instance (NSSI) concept
for granular management of network domains (e.g., radio access network (RAN), core network (CN), transport network (TN)) and provides different levels of exposure to the tenants.
Briefly, a managed NSI can consist of several NSSIs.
The NSSIs can be chosen by the network manager based on geographical circumstances,
network and access technology (e.g., a mmWave NSSI, 4G core NSSI),
vendor differences (e.g., parts of network resources are from vendor X and parts of network resources are from vendor Y),
and other factors that are significant from the network management perspective.
In this study it is assumed that NSSIs are generated based on network domain and geographical region.

To achieve the potential advantages of network slicing, network operators need to address several challenges,
such as providing QoS guarantees for different network slices/services, while efficiently utilizing the capacity of the infrastructure network.
Fulfilling such requirements mainly depends on the underlying resource provisioning mechanism that is used for resource management
and placement of virtualized network functions (VNFs) \cite{zhang2017network, zhou2016network, clayman2014dynamic}.
There are several work in the literature which study resource allocation to VNFs in the context of a general topology network.
Indeed, the VNF resource allocation problem can be traced-back/reduced to the virtual network embedding (VNE) problem,
wherein a virtual network is embedded on the top of a substrate network~\cite{fischer2013virtual}.
However, it usually results in a mixed integer linear program (MILP) which is shown to be NP-hard~\cite{amaldi2016computational}.
Hence, different heuristic methods are proposed to address VNF placement and resource provisioning in the network of an
infrastructure service provider (ISP) or an enterprise network~\cite{ghaznavi2015elastic, addis2015virtual, zhang2017joint,
 bari2015orchestrating, addad2018towards}. The reader may refer to Section~\ref{sec:related_work} for a more detailed literature survey.

By capturing the underlying resource model in a real mobile network environment, leveraging the possibility for multi-path connectivity,
and by exploiting the constraints that are imposed for placement of VNFs in practice,
we propose a novel formulation for end-to-end resource provisioning,
which avoids intractable complexities of solving an MILP.
Particularly, given the placement constraints, there remains a limited number of nodes in each domain, and therefore a limited number of paths which can provide service to NSIs in a certain region.  By considering pre-determined paths\footnote{\footnotesize It is assumed that the paths are a-priori determined by running a path finding algorithm. The detailed implementation of such an algorithm, however, is out of the scope of this paper.}
 (each comprising a pre-determined chain of VNFs) and exploiting the possibility for multi-path connectivity, we avoid intractable complexities of solving an MILP (for placement of VNFs),
while yet providing the flexibility to optimize routing across different paths/chains-of-VNFs.
The formulated problem indeed relaxes the constraint of single path routing of MILPs by exploiting the possibility for multi-path connectivity in next generation networks. A detailed description of the system model is presented in Section~\ref{sec:SysModel}.
 The proposed solution presents a \emph{market equilibrium} approach which
 results in an agile and distributed mechanism for end-to-end resource provisioning to NSIs in a \emph{multi-resource multi-domain}
 mobile network environment (such as an integrated terrestrial-aerial-satellite network).
{The proposed market-based solution best resembles the real-world interaction between the infrastructure service providers (in a vertical heterogeneous network~\cite{alzenad2019coverage}) and virtual mobile network operators (i.e., tenants) which acquire resources to implement different NSIs.
Market-based mechanisms have recently received considerable attention~\cite{luong2017resource, nguyen2018price, halabian2019distributed},
 since they may lead to a desirable performance in terms of
 resource utilization and energy-efficiency, in addition to maximizing social welfare and user satisfaction.}
Of course the distributed implementation comes at the price of a signaling overhead to exchange certain information between the network slice tenants and the infrastructure providers. The required information to exchange, however, is kept minimal as it is limited only to the resource prices and the allocated resources. Such an occasional signaling (which should be performed in case of an update) may not be significant compared to the persistent measurement and monitoring signaling that is usually communicated for the sake of network management and maintenance.

{\bf Contributions:} The contributions of this paper are summarized as follows.
\squishlist
\item {\bf Problem formulation:} We propose a new formulation for \emph{inter-domain resource provisioning to network slices},
       developing a framework which unifies the allocation of different types of resources
       (including network bandwidth as well as computing resources) in a heterogeneous multi-domain environment.
      We formulate resource provisioning to network slices in such an environment as a \emph{concave maximization problem}
       which maximizes social welfare among network slice tenants, while minimizing operational expenditures (OPEX)
       for infrastructure service providers at the same time.
\item {\bf An auction-based solution:} To solve the proposed resource provisioning problem, we devise an iterative auction game among network slice tenants, each bidding for different resources so as to maximize a local \emph{payoff} function. The infrastructure service providers (owning data centers or access point (APs) across different domains), on the other hand, decide on the resource prices. The described game is shown to be at a Nash equilibrium if and only if it is at an optimal solution to the global concave optimization problem (which provably results in a unique optimal traffic volume for each NSI).
\item {\bf Characterizing the solution:}
    It is shown that the proposed approach results in a \emph{distributed privacy-saving} mechanism which does not require sharing any private information  (e.g., resource capacities of data centers or APs, and demand profile or payoff function of tenants) among different parties.
      We further analyze the performance of the proposed mechanism, by \emph{demonstrating certain properties} (such as \emph{envy-freeness}
      and \emph{sharing incentive}) that are deemed desirable for \emph{efficient and fair} allocation of resources.
\item {\bf Demonstrating the performance:} In addition to \emph{analytical results}, we also employ numerical evaluations to show the validity of the results, demonstrate convergence of the presented algorithm, and show the superior performance of the proposed mechanism
      (in terms of resource utilization, fairness and operational expenditures) compared to heuristics and other existing solutions.
\squishend

{\bf Organization:}
%The rest of this paper is organized as follows.
The background and related work is presented in Section~\ref{sec:related_work}.
The system model is characterized in Section~\ref{sec:SysModel}.
Our proposed distributed resource provisioning mechanism is described in Section~\ref{sec:main}.
 Some import extensions to the original formulation (such as considering budget-constrained tenants, and exploiting the capabilities at the mobile edge) are presented in Section~\ref{sec:extension}.
The numerical results on evaluating the performance of the proposed scheme are reported in Section~\ref{sec:Num_res}.
The paper is concluded in Section~\ref{sec:conclusion}.

\section{Background and Related Work}\label{sec:related_work}

\subsection{Background on Enabling Technologies}
Network slicing is a promising solution for the next generation mobile networks which allows to build multiple logical
 networks on the top of a shared infrastructure, so that (virtual) mobile network operators may provide services tailored for
different network slices with different QoS requirements.
To achieve potential advantages of network slicing, considerable research activities are dedicated to developing
the underlying/enabling technologies (such as, NFV, VMs, containers, etc.) that are required for implementation of virtualized network functions and services~\cite{zhang2016flurries,martins2014clickos, palkar2015e2, zhang2016opennetvm, bremler2016openbox}. NFV decouples the software implementation of  network functions from the underlying hardware. Hence, different network appliances and/or middle-box processings (such as firewalls, traffic shapers, etc.) can be implemented on an VM running on commercial off-the-shelf hardware (such as general purpose server, storage and switches), as opposed to dedicating specialized hardware devices for implementation of network functions/protocols in conventional communications networks~\cite{han2015network, herrera2016resource}. One of the main advantages of \emph{virtualized network functions} is programmability, so that future changes can be applied easily by just updating the software without the need to replacing the hardware~\cite{herrera2016resource}.

Another line of research in this area includes the works investigating complementary technologies such as CRAN (cloud/cetralized radio access network~\cite{checko2015cloud}), mobile edge computing (MEC), and fog computing~\cite{taleb2017multi, nguyen2018price}. CRAN is a new architecture which is introduced as a cost-efficient solution to address the scalability issue with the growing user demands in 5G mobile networks~\cite{checko2015cloud}. The main idea is to achieve multiplexing gain by pooling baseband units (BBU) from several base stations into a centralized radio access unit. It introduces substantial savings on both operational expenditures (due to enhanced energy-efficiency) and capital expenditures for implementation of BBUs. Moreover, it may improve the network performance by providing the possibility to perform joint processing of signals from different base station~\cite{checko2015cloud}.
Mobile edge computing, on the other hand, is an emerging platform which integrates new technologies, such as cloud computing and NFV, with the conventional telecommunication networks to provide computational capabilities at the mobile edge, enabling a wide range of new applications/services \cite{taleb2017multi}. 
In \cite{nguyen2018price} a market equilibrium approach has been proposed to efficiently allocate the resources at the mobile edge to budget constrained users.

\subsection{Related Work}
Different business models and architectural solutions have been presented for application of NFV and network slicing to wireless/mobile networks \cite{zhang2017network, zhou2016network, liang2015wireless}. In \cite{zhou2016network} the authors propose the concept of hierarchical network slice as a service, enabling the operators to provide customized end-to-end (E2E) cellular network as a service. In \cite{zhang2017network} an architecture is presented for RAN virtualization in network-slicing-based 5G networks.
Moreover, it discusses how to address different challenges that are involved in RAN virtualization (such as power control, channel allocation and mobility management) in the proposed architecture. Indeed, the concept of network slicing can be applied to different domains (including RAN, backhaul and core network) individually \cite{nikaein2015network, costa2013radio, RAN_slicing, bagaa2014service, bagaa2018efficient}, or providing an E2E network slice as a service~\cite{halabian2019distributed}. 

To achieve potential advantages of network slicing and to provide E2E QoS guarantees for different network slices with diverse E2E QoS requirements, an efficient E2E resource provisioning mechanism is required.
There are several works in the literature which study the problem of E2E resource provisioning to network slices in the context of a general topology network (e.g., an ISP or an enterprise network). The work in \cite{clayman2014dynamic} proposes a high-level E2E orchestration framework wherein the problem is broken to placement of virtual machines across the network, and then resource allocation to network functions on virtual nodes. However, it does not provide a technical solution for placement and resource allocation to VNFs. The work in \cite{bari2015orchestrating} formulates placement of chain of VNFs as an MILP, and then proposes a dynamic-programming-based heuristic which achieves a sub-optimal solution. The complex network theory is used in \cite{guan2018service} for ranking
the nodes of an infrastructure network, and then mapping them to VNFs.
A heuristic solution is proposed in \cite{zhang2017joint} for joint VNF placement and online request scheduling to the instantiated VNFs.
In \cite{addis2015virtual}, a multi-objective optimization problem (which minimizes links traffic along with the number of busy CPU cores) is formulated for optimal routing in the network of an ISP.
In \cite{ghaznavi2015elastic} a heuristic method is proposed for placement of elastic network functions, while it minimizes the operational expenditures in order to address the trade-off between bandwidth and host resource consumption.
Indeed, all of these studies and also the work in~\cite{addad2018towards} can be viewed as an extension/variant of \emph{virtual network embedding} problem,
resulting in an MILP which is shown to be NP-hard~\cite{amaldi2016computational}.
Hence, all of these studies come up with a heuristic solution.
Moreover, the underlying model in none of these work is comprehensive in the sense to consider an accurate resource model for data centers (comprising multiple types of resources with heterogeneous resource capacities).

The work in \cite{halabian2019distributed} is the most related study to the framework presented herein.
However, it does not provide the possibility to exploit different paths/chains-of-VNFs.
Indeed, the solution in \cite{halabian2019distributed} is based on a model assuming fixed placement of VNFs over a pre-determined path for each network slice. Moreover, it does not address cost-aware resource pricing and minimization of the OPEX.
Lastly, as we show in Section~\ref{sec:Num_res}, the solution of \cite{halabian2019distributed} may not satisfy some desirable fairness related property (such as sharing incentive). 
Adopting a market equilibrium approach, we develop a distributed mechanism for E2E resource provisioning to different network slices. The proposed mechanism is shown to maximize social welfare among network slice tenants, while minimizing the OPEX across different domains.
In our formulation, we consider a number of possible forwarding paths (each comprising a pre-determined chain of VNFs)
for each network slice. In this way, we avoid intractable complexities in placement of VNFs,
while yet providing the flexibility to optimize routing across different paths/chains-of-VNFs.
Moreover, we consider an accurate resource model (comprising multiple types of resources) for data centers.

\section{System Model}\label{sec:SysModel}
Based on the developments in 5G standardization, for
each end-to-end NSI $n$, $n = 1, 2,\cdots,N$, one may consider three
NSSIs in different domains, namely, radio access, backhaul,
and the core network. % (see Fig.~\ref{fig:systemModel}).
Fig.~\ref{fig:systemModel} shows a sample NSI comprising RAN, CRAN\footnote{CRAN is actually a part of 5G core network, based on 5G standardization by SA2 group. In this article, we use the terms ``CRAN" and ``core" to distinguish between the domains where these functionalities are implemented.}, and core NSSIs.
Each NSSI comprises a specific sequence of VNFs in a certain domain.
 Note that it is possible to consider several NSSIs in each domain.
In addition, there might be more than one service provider node (i.e., data center or AP) in each domain.
The resources (including network bandwidth and computing resources) provided by a data center or AP in a particular domain
are used to implement network functions for NSSIs in the same domain. Note that virtualization helps to tailor the network function chains
with respect to the services provided by different NSSIs. Hence, each data center or AP has to implement several network
function chains, depending on the network services provided by NSSIs.

\begin{figure*}[h!]
\centering
\includegraphics[width = 0.94\textwidth]{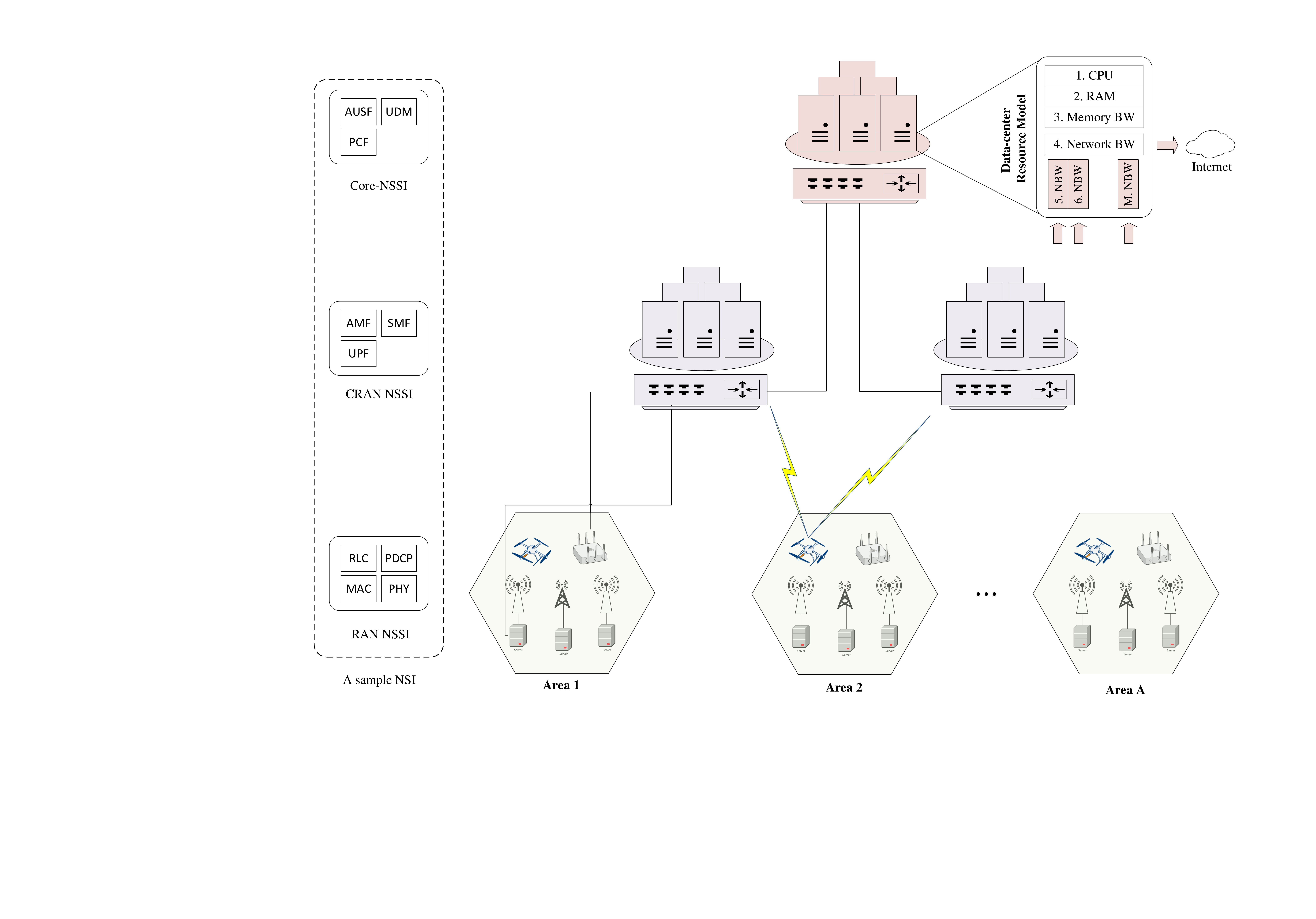}
%\footnotesize
\caption{\footnotesize End-to-end network slice virtualization in a multi-domain mobile network environment.
Sample paths from the radio access to the core network domain are shown in the figure.
%A sample NSI with VNF chains implemented by each domain is shown on the left side.
Each AP or data center is in charge of allocating different types of resources
(including the inbound network bandwidth (NBW)) to VNFs of the NSSIs in each domain.
The resource model for a generic node is shown on the top-right corner.}
\label{fig:systemModel}
\end{figure*}

Given the deployed infrastructure, it is assumed that the traffic flow in each area, $l=1,2,...,A$, can be forwarded
over certain paths towards the core network~\cite{halabian2019distributed, rost2016mobile}.
Particularly, mobile users in each geographical area could be served by a number of local radio access units.
Then, the traffic flow from each radio access unit could be forwarded over pre-defined paths towards the core network.
We assume that the set of forwarding paths are pre-defined based on VNF placement constraints
(including vendor/technology compatibility, transmission latency constraints, etc.) for each geographical area.
Let $\mathcal{P}_l$ denote the set of paths originating from area $l$.
Each path $p\in\mathcal{P}_l$ originating from area $l$ comprises a sequence of \emph{nodes}\footnote{A node at the RAN represents an access-point/base-station, whereas a node at backhaul or core network may represent a server or a data center.} over different domains of the network, $p=\{l,  i, j, k\}$, $i\in\mathcal{I}_1$, $j\in\mathcal{I}_2$, $k\in\mathcal{I}_3$, where $\mathcal{I}_1$, $\mathcal{I}_2$, and $\mathcal{I}_3$ denote the set of node indices over RAN, CRAN, and core network, respectively.
A communication link with a certain nominal capacity is assumed between every two consecutive nodes of a path.
It is assumed that the nodes in each domain are in charge of allocating the inbound communication bandwidth.
In general, we use a capacity vector of size $M$, ${\bf C}_i=[C_{i,1},...,C_{i,r},...,C_{i,M}]$, to specify the amount of available resources,
such as CPU, RAM, memory bandwidth, the outbound bandwidth towards the internet, and the inbound communication bandwidths at each node $i$.
The resource model for a generic node is shown on the top-right corner of Fig.~\ref{fig:systemModel}.

We use the notion of demand vector~\cite{DRF}, ${\bf d}^p_{n,i}=[d_{n,i,r}^p]$, to specify the amount of different resources, $r = 1, 2, \cdots, M$, that are required for processing one unit of traffic for NSI $n$ when routed to node $i$ through path $p$.
The reader may note that ${\bf d}^p_{n,i}$ is node-dependent,
so that it may reflect a variable performance for the same NSI over different nodes.
For instance, at the RAN domain the number of physical resource blocks that are required for transmitting one unit of traffic for NSI $n$
might be different from one radio access point to another~\cite{RAN_slicing}. It is worth noting that the demand vector convention follows a physical-layer abstraction model to capture mid-term statistics of the RAN (averaged over few minutes, e.g., to absorb short-term variations), while yet providing the possibility to capture some physical layer complexities (e.g., variations that may occur due to UE-mobility).
 Moreover, the demand vector is slice-specific, so it can account for variable performance of different NSIs at the same node.
 This can be particularly useful to capture slice dependent complexities at the RAN.
The demand vector, ${\bf d}_{n,i}^p$, is also considered to be \emph{path dependent} so as to account for
(possibly) different chains of network functions over different paths. Also, it may account for multiple inbound communication links at each node (including the RAN).

As an example, consider a node (i.e., server) comprising 16 CPU cores, 32 GBytes of RAM, 1 Gb/s memory BW, 10 Gb/s outbound BW,
and 4 input ports each with a bandwidth of 2.5 Gb/s, complying with the resource model in Fig~\ref{fig:systemModel}.
The capacity vector here is represented as ${\bf C}_i=[16,32,1,10,2.5,2.5,2.5,2.5]$.
A demand vector, for instance, can be described\footnote{In practice, a hypervisor shares the CPU-time among different VMs running on the same server. So, any fraction of a CPU core can be assigned to each VM~\cite{DRF}.} as ${\bf d}^p_{n,i}=[0.5,2,0.1,0.75,1.1,0,0,0]$, which specifies the amount of different resources
required for processing one unit of traffic (e.g., 1 Gb/s) for NSI $n$ when routed through the first input port to node $i$.
Identifying demand vectors per routing paths is the key to account for limited bandwidth of communication links.
This is in contrast to existing works in the literature which only consider the overall network bandwidth over each data center or server~\cite{halabian2019distributed, fossati2019multi, nguyen2018market, leconte2018resource}.
Such studies may only address the limited capacity of the server switch in accepting the incoming traffic,
 but may not capture the capacity constraint of communication links.

Note that one mobile user can obtain multiple services, each provided by a different NSI. % newA %so it can be associated with multiple NSIs.
The key point is that the granularity level is per slice per traffic-flow, rather than per mobile user.
Let $x_n^p$ denote the \emph{capacity} {allocated/provisioned} to NSI $n$ over path $p$.
It actually represents the (maximum) volume of traffic that can be forwarded for NSI $n$ over path $p$.
The (maximum) volume of traffic which can be forwarded for the users/tenant of NSI $n$ in area $l$ is given by
\be
x_{n,l}:=\sum_{p\in{\mathcal P}_l}x_n^p.\label{x_area}
\ee
Indeed, the traffic admitted from users of NSI $n$ should be kept below $x_{n,l}$ in each area.
The resources allocated from node $i$ to
%the assigned VM to
NSI $n$ (to provision $x_n^p$ units of traffic over each path $p$ which node $i$ belongs to) is given by ${\bf a}_{n,i}=[a_{n,i,r}]$,
\be
a_{n,i,r} = \sum_{p: i\in p}x_n^p {d}^p_{n,i,r}.\label{a_server}
\ee
Generally, for an allocation, ${\bf x}:=\{x_{n}^p\mid n\in\mathcal{N}, p\in\mathcal{P}_l, l=1,2,..,A\}$ to be feasible, the following has to be satisfied:
\be
\sum_{n\in\mathcal{N}}a_{n,i,r}=\sum_{n\in\mathcal{N}}\sum_{p:i\in p} x_n^pd^p_{n,i,r}\le C_{i,r},~\forall i,r.
\ee

\begin{table}
\vspace{+2mm}
\footnotesize
\caption{A summary of key notations}
\label{table3}
%\centering
\begin{tabular}{| l | l |}%{p{8.48cm}}
\hline
Notation & Description\\
\hline\hline
$C_{i,r}$                           & Capacity of resource $r$ at node $i$\\\hline
${\bf d}^p_{n,i}=[{d}^p_{n,i,r}]$   & The demand vector for NSI $n$ at node $i$ using path $p$\\\hline
${\bf a}_{n,i}=[a_{n,i,r}]$         & Vector of resources allocated from node $i$ to NSI $n$\\\hline
$x_n^p$                             & The traffic capacity provisioned to NSI $n$ over path $p$\\\hline
$\mathcal{P}_l$                     & The set of paths originating from area $l$\\\hline
$x_{n,l}$                           & The traffic capacity provisioned to NSI $n$ in area $l$\\\hline
$U_{n,l}(x_{n,l})$                  & The utility (function) for NSI $n$ in area $l$\\\hline
$q_{i,r}$                           & The OPEX for unit of resource $r$ at node $i$.\\\hline
$\mu_{i,r}$                         & The price to book one unit of resource $r$ at node $i$\\\hline
$\Pi_n({\bf x}_n; \mu)$               & The net payoff function for NSI $n$\\\hline
$w_{n,i,r}$                         & The payment of NSI $n$ to node $i$ for resource $r$\\\hline
\end{tabular}
\end{table}

 The key notations are summarized in Table~\ref{table3}.
It is worth noting that the resource capacity, $C_{i,r}$, represents the \emph{nominal capacity} of resource $r$ at node $i$. For example, the communication bandwidth of an AP is defined as the maximum achievable data rate when using the best modulation and coding scheme. Depending on their performance requirements and to rectify certain physical layer complexities, however, some NSIs may need to use lower order modulation and coding schemes. Such requirements can be flexibly reflected in the demand vectors.

\section{Distributed Resource Provisioning}\label{sec:main}
In this section, we first formulate the problem of resource provisioning to different network slices as a \emph{centralized system-wide} optimization problem, maximizing social welfare among network slice tenants, while minimizing OPEX for the infrastructure service providers at the same time.
We show that the described problem can be solved by implementing an auction game among network slice tenants on one hand, %each maximizing a local \emph{payoff} function.
and a set of price-taking infrastructure service providers, on the other hand.
In the proposed solution, each network slice tenant maximizes a local \emph{payoff} function,
while a certain resource allocation and pricing scheme is implemented at each infrastructure node.
%We show that there exists a unique Nash equilibrium for such an auction game. Moreover,
We further characterize the performance of the proposed mechanism by demonstrating
certain properties which are highly desirable for efficient and fair allocation of the resources.

\subsection{The System-Wide Objective}
Each network slice tenant may gain a utility (i.e., \emph{profit}) of $U_{n,l}(x_{n,l})$ out of the allocated capacity in each area $l$.
The utility function for each NSI $n$ can be reasonably represented by a \emph{concave function}. %\footnote{In our technical report \cite{DRP} we discuss how to incorporate the end-to-end delay into NSI's utility function.}~\cite{kelly1998rate}.

\vspace{+2mm}
\noindent{\bf Problem~1: The system-wide optimization problem}
\be
&& \max_{\bf x} \sum_{n,l}U_{n,l}\left(\sum_{p\in\mathcal{P}_l}x_n^p\right)-\sum_{i,r}q_{i,r}\sum_n\sum_{p:i\in p}x_n^pd^p_{n,i,r}\qquad\label{global_obj}\\
&& \text{s.t.} \sum_n\sum_{p:i\in p}x_n^pd^p_{n,i,r}\le C_{i,r}, ~~\forall i,r,\label{global_c1}\\
&& \text{~~~~} x_n^p\ge0, ~~\forall n,p\label{global_c2}.
\ee
The parameter $q_{i,r}$ is the OPEX for node $i$ to provision one unit of resource $r$.
So the objective in Problem~1 is to \emph{maximize the overall utility} for different NSIs (i.e., \emph{maximizing social welfare}),
while minimizing OPEX for service providers. The following solution method and the presented analytical results are established
\emph{in general for every utility function}, $U_{n,l}(\cdot)$, that is continuously differentiable and strictly concave.
As an example, one may assume $U_{n,l}(z)$ from a common class of utility functions for which the derivative (i.e., \emph{marginal benefit}) is described as $U'_{n,l}(z)=(\phi_{n,l}/z)^{\alpha_n}, ~\alpha_n>0$ ~\cite{BG92, kelly1998rate}.
For this class of functions, the parameter $\alpha_n$ (which typically takes on a value in the range of $[1,2]$) determines the shape of the utility function~\cite{BG92, kelly1998rate}.
The parameter $\phi_{n,l}$ may present the traffic load for each NSI $n$ in a certain area.
We use the class of functions for the sake of numerical evaluations in Section~\ref{sec:Num_res}.

For NSIs with delay sensitive applications, we can account for the end-to-end delay in the utility function. %the payoff may also depend on the achievable delay.
In particular, the end-to-end delay for NSI $n$ which forwards an \emph{admissible}\footnote{The admitted traffic volume could be a fraction of offered load in each region. In Section~\ref{sec:extension:BC}, we discuss a potential solution to adjust the admitted traffic volume.} volume of $\phi_{n,l}$ unit of traffic can be estimated by~\cite{ETEDelay5G}
\be
D_n(x_{n,l}):= \frac{L}{x_{n,l}-\phi_{n,l}}+\frac{h_{n,l}L}{x_{n,l}},\text{ for } x_{n,l}> \phi_{n,l},
\ee
where $L$ is the average packet length (in bits) and $h_{n,l}$ is (proportional to) the total number of steps that
each packet is processed in the network\footnote{By definition, $D_n(x_{n,l}):=-\infty$ for $x_{n,l}\le\phi_{n,l}$.}.
The net revenue for NSI $n$ then is given by
\be
R_n(x_{n,l}) := U_n(x_{n,l}) - \beta_{n,l}D_n(x_{n,l}),\label{eq_revenue}
\ee
where the parameter $\beta_{n,l}$ in \eqref{eq_revenue} relates the delay to loss in profit.
For delay sensitive slices, we can consider the net revenue instead of the utility function in \eqref{global_obj}.
The parameter $\beta_{n,l}$ then could be properly adjusted to keep the end-to-end delay less than a desired threshold.
The proposed solution method in the following can be viewed as a variant of
the market equilibrium approach which was originally presented in \cite{kelly1998rate} and here is extended to a multi-resource multi-domain network.

\subsection{End-to-end Network Slice Management}\label{sec:main:NSP}
For each NSI we assume an end-to-end slice manager which monitors
the end-to-end network slice performance, and decides on bidding for the resources at different domains of the network.
It is assumed that the network slice manager has (a-priori) found the demand vector for each service function chain of an NSI,
by getting feedbacks and monitoring the resource utilization of the VMs implementing the function chains.
Particularly, assume that $\hat{\bf d}_{n,i} = [\hat{d}_{n,i,r}]$ is an (arbitrary) initial estimate of the demand vector for a service function at node $i$. Given that the resources are initially allocated (to the service function chain) proportional to $\hat{\bf d}_{n,i}$, the resource utilization of the VM, which implements it, is described in terms of $\hat{\bf d}_{n,i}$ and the true value of the demand vector~\cite{khamse2017per},
\be
u_{n,i,r}:=\frac{d_{n,i,r}}{\hat{d}_{n,i,r}}\min_{r'}\{\frac{\hat{d}_{n,i,r'}}{{d}_{n,i,r'}}\}.\label{utilization_Res}
\ee
So, given the resource utilization and an (\emph{arbitrary}) initial estimate of the demand vector, the true value of the demand vector can be found using \eqref{utilization_Res}. In the following, it is assumed that the true value of the demand vector is known to the slice manager.
Based on the demand vector, the slice manager knows the amount of resources that are required for forwarding/processing a certain volume of traffic at each domain of the network.
Moreover, it is assumed that the utility function (as a function of NSI traffic volume) is known to the network slice manager.
Each network slice manager then needs to maximize its net \emph{payoff function} (i.e., the gained utility minus the total payment),
\be
\Pi_n({\bf x}_n; {\bf \mu}) = \sum_lU_{n,l}(x_{n,l}) - \sum_{i,r}\sum_{p:i\in p}x_n^p d^p_{n,i,r}\mu_{i,r},\label{eq_payoff}
\ee
where $\mu_{i,r}$ is the price to book one unit of resource $r$ at node $i$.
The reader may note that $x_{n,l}$ is described in terms of path traffic volumes, $x_n^p$ (see \eqref{x_area}).
The network slice manager then strives to find an allocation ${\bf x}_n:=\{x_n^p\mid p\in\mathcal{P}_l,l=1,2,...,A\}$ which solves the following problem.

\vspace{+2mm}
\noindent{\bf Problem 2:  Network slice optimization problem}
\be
%\text{Problem 1: }
&&\max_{{\bf x}_n}\Pi_n({\bf x}_n; {\bf \mu})\qquad\qquad\\
&&\text{Subject to: }x_n^p\ge 0.
\ee
Given a solution to this problem, the network slice manager decides how much capacity should be provisioned for each NSI $n$ over each path,
which subsequently specifies the amount of resources which should be allocated from each node to NSI $n$ (see \eqref{a_server}).
The bid/payment that is made by NSI $n$ for resource $r$ of node $i$ is given by
\be
w_{n,i,r}:= \mu_{i,r}\sum_{p:i\in p}x_n^p d^p_{n,i,r}.\label{eq_bids}
\ee
\begin{note}
The payment matrix for each NSI $n$, ${\bf W}_n:=[w_{n,i,r}]$, is determined based on path traffic volumes/capacities.
That is, ${\bf W}_n={\bf W}_n({\bf x}_n)$.
\end{note}

\subsection{Resource Management for Nodes}\label{sec:main:SNSP}
Each service provider node, which implements a \emph{subnetwork functionality} (i.e., NSSI) for NSI $n$,
receives some bids $w_{n,i,r}>0$ for different resources.
It is assumed that each node does not price discriminate among different NSIs.
That is, each node $i$ chooses a certain price, $\mu_{i,r}$, for unit of each resource $r$,
and subsequently allocates an amount of $a_{n,i,r}=w_{n,i,r}/\mu_{i,r}$ of resource $r$ to each NSI $n$.
It is assumed that node $i$ is incurred an OPEX of $q_{i,r}$ for providing one unit of resource $r$.
This may include electrical energy costs, and the cost to transport traffic over the network of an internet service provider.
To ensure that OPEX are covered by NSI payments, the price for one unit of resource $r$ is chosen to be\footnote{We consider a competitive market in the presence of a plurality of service provider nodes. So it is assumed that each subnet service provider node applies the true value of its OPEX to set the resource prices.}  %The extension where service providers make strategic decisions by altering their OPEX is out of the scope of this paper.}
\be
\mu_{i,r} = \max\left\{q_{i,r}, \frac{\sum_nw_{n,i,r}({\bf x})}{C_{i,r}}\right\},\label{res_price}
\ee
so that resource $r$ is \emph{cleared} (that is $\sum_na_{n,i,r}=C_{i,r}$) only when ${\sum_nw_{n,i,r}}/{C_{i,r}}\ge q_{i,r}$, %$\mu_{i,r} \ge \max\{q_{i,r}$.
In general,
\be
\sum_na_{n,i,r}=\sum_n\frac{w_{n,i,r}}{\mu_{i,r}}=:\eta_{i,r}C_{i,r},\\
\text{where }\eta_{i,r}:= \min\left\{\frac{\sum_nw_{n,i,r}}{C_{i,r}q_{i,r}},1\right\},
\ee
so that a fraction of $\eta_{i,r}<1$ of resource $r$ is allocated to different NSIs when ${\sum_nw_{n,i,r}}/{C_{i,r}}<q_{i,r}$.
 We show that employing such a simple pricing scheme in conjunction with the end-to-end NSI resource management mechanism of Section~\ref{sec:main:NSP} results in minimizing the operational costs in the whole network, while maximizing social welfare among network slice tenants.

\subsection{Distributed Online Mechanism}
%\subsection{Distributed Online Resource Allocation Mechanism}
In this section we devise a distributed online mechanism which implements
an auction game among the network slice mangers, on one hand, and a set of price-taking infrastructure/subnet service providers, on the other hand.
In the proposed mechanism, the network slice managers (each implementing one NSI) may iteratively update their bids for various resources of different nodes (see \eqref{eq_bids}). Given the payments by network slice managers,
each node (i.e., AP or data center) decides on the resource prices (see \eqref{res_price}) and
declares\footnote{Alternatively, we may assume that nodes declare the allocated resources to each NSI,
so that slice managers indirectly infer the resource prices ($\mu_{i,r}=w_{n,i,r}/a_{n,i,r}$).}
the resource prices to all network slice managers.
Each network slice manager then may find its desired (i.e., optimal) path traffic volumes,
and accordingly updates its payments, so as to improve its achievable payoff.
Given the payments from all NSIs, the resource prices are updated in a way that the capacity constraints are met for all resources.
This procedure continues until no network slice is willing to update its payments.
At this point (which is so-called a \emph{Nash equilibrium}) a stable price is established for each resource.
Since the resource prices are determined based on the payments from all NSIs, a tradeoff (which results in maximizing social welfare, as shown in Theorem~\ref{th_nash_opt}) is established between the allocated resources to different NSIs.

In the described auction game, the bidders are network slice managers, while the suppliers are different nodes,
each offering $M$ types of (divisible) resources according to a certain pricing scheme (as in \eqref{res_price}).
Hence, the resource prices are function of actions (i.e., allocations) taken by network slice managers (see \eqref{res_price} and Note~1).

\begin{definition}\label{def_NE}
The game is said to be in a \emph{Nash Equilibrium (NE)} if there exists an allocation ${\bf x}^*$ and a set of %set of payments ${\bf w^*}={\bf w}({\bf x}^*)$
resource prices ${\bf\mu}^*=\mu({\bf x}^*)$, in compliance with \eqref{res_price}, so that no NSI gains payoff by \emph{unilateral deviation} from its current allocation,
\be
\Pi_n({\bf x}^*_n; {\bf \mu}^*)\ge \Pi_n({\bf x}_n; {\bf \mu}^*), \text{ for all feasible }{\bf x}_n, \forall n.\label{NE_condition}
\ee
\end{definition}

It is worth noting that the payoff of each NSI depends on its own \emph{action} (i.e., ${\bf x}_n$) as well as \emph{the current resource prices}.
The inequality in \eqref{NE_condition} implies that ${\bf x}^*_n$ is an optimal solution to Problem~2 (i.e., the local payoff optimization for NSI $n$) in conjunction with the equilibrium resource prices. Given the concavity of the payoff functions, an allocation
${\bf x}_n:=\{x_n^p\mid p\in\mathcal{P}_l,l=1,2,...,A\}$ is an optimal solution to Problem~2 for NSI $n$ if and only if for every path $p\in\mathcal{P}_l$:
\be\nonumber
\frac{\partial\Pi({\bf x}_n; {\bf \mu})}{\partial x_n^p}=\qquad\quad\qquad\qquad\qquad\qquad\quad\\
 U'_{n,l}(x_{n,l}) - \sum_{i\in p}\sum_rd^p_{n,i,r}\mu_{i,r}
\begin{cases}
= 0    & \text{if }x_{n}^p>0,\\
\le 0  & \text{otherwise}.
\end{cases}\label{slice_opt_cond}
\ee
The condition in \eqref{slice_opt_cond} implies that $x_n^p>0,~p\in\mathcal{P}_l$, only when
\be
\sum_{i\in p}\sum_rd^p_{n,i,r}\mu_{i,r} = \min_{p'\in\mathcal{P}_l} \sum_{i\in p'}\sum_rd^{p'}_{n,i,r}\mu_{i,r}.\label{path_optimality}
\ee
The left hand side in \eqref{path_optimality} gives the cost for transmitting one unit of NSI $n$'s traffic over path $p$.
It means that at the optimal solution to Problem~2, the traffic for each NSI is forwarded over
the least expensive path(s) in each area. Moreover, it follows from \eqref{slice_opt_cond} that %the optimum traffic volume is set to
\be
{x_{n,l}^{*}}:={U'_{n,l}}^{(-1)}\left(\min_{p\in\mathcal{P}_l} \sum_{i\in p}\sum_rd^{p}_{n,i,r}\mu_{i,r}\right),\label{optimal_volume}
\ee
where ${U'_{n,l}}^{(-1)}(\cdot)$ is the \emph{inverse} of ${U'_{n,l}}(\cdot)$.
The function ${U'_{n,l}}(\cdot)$ is assumed to be \emph{invertible}
owing to concavity of $U_{n,l}(\cdot)$ over its \emph{feasible region}.
Although ${x_{n,l}^{*}}$ is uniquely specified for each NSI in each area,
yet there might be several possible allocations in case that the minimum transmission cost is attained over multiple paths.
In particular, let $\mathcal{P}^*_{n,l}\subseteq\mathcal{P}_l$ denote the set of paths which result in the minimum transmission cost
for NSI $n$ in area $l$ (c.f. \eqref{path_optimality}).
A possible allocation is to uniformly distribute $x_{n,l}^{*}$ across the least expensive paths:
\be
x_n^{p*}=\begin{cases}
x_{n,l}^{*}/|\mathcal{P}^*_{n,l}|   &  \text{if }p\in\mathcal{P}^*_{n,l},\label{optimal_allocation}\\
0   &  \text{otherwise}.
\end{cases}
\ee

The proposed Distributed Resource Provisioning (DRP) mechanism is summarized in Table~\ref{table1}.
Beginning with some initial resource prices (e.g., $\mu_{i,r}=q_{i,r}$),
the slice manager for each NSI may find an optimal allocation according to
\eqref{optimal_volume} and \eqref{optimal_allocation}.
However, to prevent oscillations, especially when traffic is to be distributed over multiple paths,
${\bf x}_n$ is (gradually) updated according to \eqref{update_rule}.
The slice manager then finds the amount of resources which should be allocated from different nodes
(c.f. \eqref{a_server}), and accordingly bids for different resources of each node.
The offered payments are in turn used to update the resource prices at each node (see \eqref{res_price}).
Let $\{\hat{\mu}_{i,r}\}$ denote the updated resource prices which are taken by node $i$ in response to bids made by different NSIs in the current iteration. The \emph{actual volume} of traffic that can be processed for NSI $n$ over path $p$ is given by
\be
\hat{x}_n^p := x_n^p\min_{i\in p, r}\frac{\mu_{i,r}}{\hat{\mu}_{i,r}}.\label{actual_traffic}
\ee
The updated resource prices are subsequently used by network slice managers to repeat the same procedure in the next round.
The slice managers keep updating their decisions while $\|{\bf \hat{x}}_n-{\bf x}^*_n\|>\epsilon$ for some $\epsilon>0$.

\begin{table}
\vspace{+2mm}
\footnotesize
%\captionsetup{labelformat=empty}
\caption{Distributed Resource Provisioning (DRP) Mechanism}
\label{table1}
%\centering
\begin{tabular}{p{8.48cm}}
\hline\noalign{\smallskip}
The resource prices are initially set to $\mu_{i,r}=q_{i,r}, \forall i, r$.
\begin{itemize}
\item [I.] {\bf Given a set of resource prices, ${\bf \mu}$, each network slice manager}
  \begin{itemize}
    \item [-] Finds the optimal allocation ${\bf x}^*_n$ per the current resource prices according to \eqref{optimal_volume} and \eqref{optimal_allocation}.
    It then updates the current allocation according to
    \be
       {\bf x}_n \leftarrow (1-\eta_n){\bf x}_n + \eta_n{\bf x}^*_n,\label{update_rule}
    \ee
    where $\eta_n\in(0,1)$.
    \item [-] Finds the amount of resources which should be allocated from each node (see \eqref{a_server}), and then bids for different resources, accordingly (see \eqref{eq_bids}).
  \end{itemize}
\item [II.] {\bf Given updated payments by slice managers, each service provider node}
  \begin{itemize}
    \item [-] Updates the resource prices according to \eqref{res_price},
              and allocates the resources to different NSIs, accordingly.
    \item [-] Declares the resource prices to the network slice managers.
  \end{itemize}
\item [III.] {\bf Given updated resource prices, $\hat{{\bf\mu}}$, each network slice manager}
  \begin{itemize}
    \item [-] Finds the actual allocation, $\hat{{\bf x}}_n$, according to \eqref{actual_traffic}.
    \item [-] Updates the resource prices, ${\bf \mu}\leftarrow\hat{{\bf\mu}}$, as well as the allocation, ${{\bf x}}_n\leftarrow\hat{{\bf x}}_n$. Subroutine I then is repeated while $\|{\bf {x}}_n-{\bf x}^*_n\|>\epsilon$.
  \end{itemize}
\end{itemize}\\
\noalign{\smallskip}\hline
\end{tabular}
\end{table}

\subsection{Characterizing the Solution}\label{sec:main:ChSol}
In this section we show that the proposed distributed resource provisioning mechanism results in
optimizing the global system-wide objective of Problem~1, maximizing social welfare
among the network slice tenants while minimizing the OPEX for service providers.
Moreover, we show that Problem~1 has a unique optimal solution (in terms of $\{x_{n,l}\}$)
and so is the NE of the DRP mechanism.
We finally study the properties of the resulting allocation by characterizing the solution of Problem~1.

\begin{theorem}\label{th_nash_opt}
Assume a continuously differentiable and strictly concave utility function, $U_{n,l}(\cdot)$, for each NSI $n$.
An allocation ${\bf x}$ is an NE for DRP mechanism (in conjunction with some resource prices)
if and only if it is a solution to Problem~1.
\end{theorem}

The following is a direct conclusion of Theorem~\ref{th_nash_opt}.
\begin{remark}
There exists an NE for the DRP mechanism.
Moreover, the resulting allocation at the NE is unique in terms of $\{x^*_{n,l}\}$.
\end{remark}

The proof appears in the Appendix.
According to Theorem~\ref{th_nash_opt}, the DRP mechanism results in an allocation which
maximizes the overall utility for different NSIs
minus the summation of OPEX for infrastructure service providers (see \eqref{global_obj}).
Moreover, it follows from the proof of Theorem~\ref{th_nash_opt} that the resource prices at the NE are associated with
the dual variables $\lambda_{i,r}$ corresponding to the capacity constraints in \eqref{global_c1}.
In particular, $\mu_{i,r}=q_{i,r}+\lambda_{i,r}$, so that resource $r$ with a restricting capacity at node $i$
(i.e., with $\lambda_{i,r}>0$) results in $\mu_{i,r}>q_{i,r}$.
%while minimizing the OPEXs for service providers (see \eqref{global_obj}).

The reader may note that the convergence behavior of the proposed DRP mechanism
depends on the actual choice of utility functions, which can be different for various NSIs.
The fact that different NSIs can take different utility functions makes it difficult to derive analytical results on the rate of the convergence.
To derive such results, one needs to perform a statistical or worst-case analysis which is not straightforward and is out of the scope of this paper.
Nevertheless, Theorem~\ref{th_nash_opt} describes the NE, which is the
convergence point of the DRP mechanism, as the optimal solution to Problem~1
for any choice of continuously differentiable and strictly concave utility functions.
 At this point, we leave this rich topic for future work, and provide
numerical experiments to have some observations on the convergence
behavior of the DRP mechanism in Section~\ref{sec:Num_res}.

To further characterize the NE of the DRP mechanism (or equivalently the solution to Problem~1),
we formulate an equivalent network-wide optimization problem, which provably results in the same allocation.
Towards that, let $w_{n,l}$ denote the total payment made by NSI $n$
for the resources which process the originating traffic from area $l$,
\be
w_{n,l}:=\sum_{i,r}\mu_{i,r}\sum_{p\in\mathcal{P}_l:i\in p}x_n^pd^p_{n,i,r}.\label{paymentPerRegion}
\ee

\begin{theorem}\label{TH_prop_fair}
Let $\{w^*_{n,l}\}$ denote the set of payments made by different NSIs when the DRP mechanism is in an NE.
An allocation ${\bf x}^*$ serves as an NE for the DRP mechanism if and only if it is a solution to following problem.

\vspace{+2mm}
\noindent{\bf Problem~3: Subnetworks optimization problem}
\be
&& \max_{\bf x} \sum_{n,l} w^*_{n,l}\log(x_{n,l})  - \sum_{i,r}q_{i,r}\sum_n\sum_{p:i\in p}x_n^pd^p_{n,i,r}\qquad\label{global_SN_obj}\\
&& \text{s.t.} \sum_n\sum_{p:i\in p}x_n^pd^p_{n,i,r}\le C_{i,r}, ~~\forall i,r,\label{global_SN_c1}\\
&& \text{~~~~} x_n^p\ge0, ~~\forall n,p\label{global_SN_c2},
\ee
where $x_{n,l}$ is written in terms of $x_n^p$ according to \eqref{x_area}.
\end{theorem}

According to Theorem~\ref{TH_prop_fair}, the DRP mechanism results in an allocation which satisfies
weighted proportional fairness among different NSIs (wherein the weights are set to payments) while minimizing the OPEX~\cite{kelly1998rate, khamse2018cost}.
To further characterize the allocation resulting from the DRP mechanism,
let ${\bf a}^p_n=\{a^p_{n,i,r}|i\in p, \forall~r\}$ denote the allocated resources to NSI $n$
over path $p$. We denote by $T_n({\bf a}_{n,l})$ the volume of traffic which can be processed for NSI $n$
using the resources allocated to NSI $n$ in area $l$, ${\bf a}_{n,l}:=\sum_{p\in\mathcal{P}_l}{\bf a}^p_n$.

\begin{definition}
An allocation is said to satisfy envy-freeness if each NSI $n$ in each area
would not prefer the allocated resources to another NSI when adjusted according to their payments,
that is, $T_n({\bf a}_{n,l})\ge T_n(\frac{w^*_{n,l}}{w^*_{m,l}}{\bf a}_{m,l})$.
\end{definition}

The envy freeness property embodies the notion of \emph{fairness}~\cite{DRF}. The other property that we consider here is \emph{sharing-incentive},
which ensures that the proposed mechanism outperforms a so-called \emph{uniform allocation}.
To find a generic uniform allocation, let $w^{p*}_{n,i}$ denote the payment made by NSI $n$ to node $i$ for the service over path $p$.
Consider a \emph{uniform allocation} wherein a fraction $w^{p*}_{n,i}/\sum_{m,p':i\in p'}w^{p'*}_{m,i}$ of different resources at node $i$ is
dedicated to NSI $n$ for service over path $p$.

\begin{definition}
An allocation is said to satisfy \emph{sharing-incentive} if each NSI is provided with more traffic volume compared to the \emph{uniform allocation}.
\end{definition}

The sharing-incentive property is the key to ensure a worst-case performance guarantee for each NSI.
Satisfying this property also may incentify different carriers/operators to pool their resources together~\cite{DRF},
because each of them may forward more traffic volumes over the shared infrastructure (when orchestrated by the proposed mechanism) compared to the case that each of them gets \emph{an equal (weighted) share} of all the resources.
We show that the envy-freeness and sharing-incentive properties %which are desirable for efficient and fair allocation of resources,
are established at the NE of the DRP mechanism.

\begin{theorem}\label{TH_properties}
The allocation at the NE of the DRP mechanism satisfies both \emph{envy-freeness} and \emph{sharing-incentive} properties.
\end{theorem}
The proof appears in the Appendix.

\section{Extensions}\label{sec:extension}
We may consider several possible extensions/applications for the original resource provisioning mechanism presented herein.
For instance, one may find the proposed mechanism particularly useful for network slicing in multi-layer networks
such as the integrated vertical HetNets~\cite{alzenad2019coverage} (with aerial BSs~\cite{8641421},  LEO satellites, etc.),
owing to agility and distributed nature of the solution.
In the following we briefly describe some important extensions to the original mechanism of Section~\ref{sec:main}.

\subsection{Exploiting the Capabilities at the Mobile Edge}
Our proposed mechanism can be easily extended to address the case that the service function chain at each domain comprises
a number of sub-chains, with the possibility that some backhaul and/or core network sub-chains
 can be implemented at a domain closer to the mobile edge.

In particular, let $\{\mathcal{F}_n^1, \mathcal{F}_n^2, \mathcal{F}_n^3\}$ denote the per domain network function chain for NSI $n$.
%at different domains of the network.
 It is assumed that for the last two domains $\mathcal{F}_n^s=\{\tilde{F}_n^s, \hat{F}_n^s\}$, $s=2,3$,
where $\hat{F}_n^s$ should be placed at a domain $s$ node, while $\tilde{F}^s_n$ can be flexibly placed at a node either in domain $s$ or $s-1$.
Accordingly, we may consider separate demand vectors, $\hat{\bf d}^p_{n,i}$ and $\tilde{\bf d}^p_{n,i}$, for each sub-chain of NSI $n$,
where ${\bf d}^p_{n,i}=\hat{\bf d}^p_{n,i}+\tilde{\bf d}^p_{n,i}$.
 To extend the DRP mechanism we may redefine paths as sequence of nodes, $p=\{l, i, j, k, g, h\}$,
 which host different sub-chains of an NSI. It reduces to the original formulation (with a per-domain service function chain), when $i_2=i_3$ and $i_4=i_5$.
The DRP mechanism then is implemented as before,
except that the network slice manager now %treats sub-chains as additional network functions, and
bids separately for individual sub-chains when they are located at different nodes.  % according to path $p$
The DRP mechanism may exploit this flexibility to efficiently utilize the capabilities at the mobile edge,
so as to improve the performance for NSIs with stringent QoS/delay requirements.

\subsection{Budget-Constrained Tenants}\label{sec:extension:BC}
% By the node resource management strategy
With the resource pricing strategy described in Section~\ref{sec:main:SNSP},
the resource prices might be chosen well above the operational expenditures (i.e., $\mu_{i,r}>q_{i,r}$)
when so many NSIs contend for the resources of the same node.
If the allocated traffic volume to an NSI is less than its desired optimal traffic volume,
the corresponding network slice manager may intend to make a larger bid/payment, which in turn may increase the resource prices.
This procedure may unboundedly increase the resource prices as well as the required payments from different NSIs.
In practice, however, each network slice tenant may have a limited budget, so that the bids may not go beyond a certain limit.
The key to account for limited budgets is to exploit an admission control policy which limits the \emph{admitted traffic volume}
in each area, so that the required payment remains in a feasible region.
Let $B_{n,l}$ denote the budget for NSI $n$ in area $l$, and $U'_{n,l}(x_{n,l})=(\phi_{n,l}/x_{n,l})^{\alpha_n}$,
where $\phi_{n,l}$ represents the admitted traffic volume.
In case that $w_{n,l}>B_{n,l}$, one may update $\phi_{n,l}\leftarrow\zeta\phi_{n,l}$, $\zeta<1$,
so as to make sure that the payments are less than or equal to the budget.

\subsection{Multi-resource Fair Allocation}\label{sec::ext::MRF}
In Section~\ref{sec:main:ChSol} we showed that the proposed distributed resource provisioning mechanism provides
\emph{weighted proportional fairness} among different NSIs, wherein the weight for each NSI $n$ in area $l$ is set to its payment, $w_{n,l}$ (see Problem~3).
The proportional fairness objective in Problem~3 is shown to satisfy some highly desirable properties,
such as \emph{envy-freeness} and \emph{sharing incentive}~\cite{khamse2018}.
It should be noted that the formulation in Problem~3 is applicable for a system where all the resources of all nodes
are allocated by a central controller. With a centralized implementation, however,
there are other properties, such as \emph{strategy proofness} which are desirable to be satisfied~\cite{DRF}.
In particular, an allocation mechanism is said to satisfy \emph{strategy proofness}
if each NSI may not be allocated more traffic volumes when lying about its resource demands to the centralized controller.

Dominant resource fairness (DRF) is the first multi-resource fair allocation mechanism which satisfies strategy-proofness
in addition to the above-mentioned properties~\cite{DRF}, when allocating multiple types of resources from a single \emph{server}.
Of all the resources requested by a network function from one node (for each unit of traffic), its dominant resource is the one with the highest demand when demands are expressed as fractions of the overall resource capacities.
Using DRF, each network function receives a fair share of its respective dominant resource~\cite{DRF}.
The studies in \cite{khamse2018cost, khamse2018, DRFH15, Ashari19a} investigate the problem of multi-resource fair allocation in an environment
of heterogeneous and geographically distributed servers.
Such studies, however, address \emph{single-hop processing} of the traffic in a network of cloud computing servers,
and may not be directly applicable to end-to-end resource provisioning in a multi-domain mobile network environment.

The study in \cite{halabian2019distributed} strives to extend DRF to a multi-domain mobile network environment,
wherein the traffic for each NSI is forwarded over a certain path towards the core network.
Towards that, it identifies an \emph{end-to-end dominant resource} for each NSI, which is specified based on its end-to-end demand vector.
Then, it strives to allocate each NSI a fair share of its respective dominant resource.
Particularly, a dominant share for NSI $n$ over path\footnote{Here, we may drop the index $p$ from the demand vectors,
since each NSI is presumably forwarded over a single path.} $p$ is defined as \cite{halabian2019distributed}
\be
z_n := x_n \max_{i\in p}~\max_r\frac{d_{n,i,r}}{C_{i,r}}.
\ee
Then the NSI traffic volumes, $\{x_n\}$, are determined so as to maximize the minimum dominant share across different NSIs.
In case that different NSIs make different payments, $w_n$, one can maximize the minimum \emph{weighted dominant share}, ${z_n/w_n}$, across different NSIs. We refer to such an extension to DRF as \emph{multi-domain DRF}.
However, as we show in the following (and also in Section~\ref{sec:Num_res}), the multi-domain DRF mechanism may not satisfy the sharing-incentive property, and also may not result in a \emph{fair} allocation.
Intuitively, the main problem with this mechanism is that it may identify a resource at a lightly loaded node
(with a few NSIs passing through) as the \emph{end-to-end dominant resource} for some NSIs.
Such a resource, however, may not serve as a \emph{bottleneck} over an end-to-end routing path.
In this case, a smaller share of the actual bottleneck resource is allocated to such NSIs,
which in turn may violate the sharing incentive property.

For instance, consider the example of Fig.~\ref{fig:counterExmpl},
which shows a simple network connecting the users from two APs to a CRAN unit.
%In this example, we assume six NSIs with a certain forwarding path for each of them.
The resource capacity vector for each node, ${\bf C}_i, i = 1, 2, 3$, as well as the demand vector for each NSI $n$ passing through node $i$, ${\bf d}_{n,i}$, are shown in the figure. In this example, the communication bandwidth of APs is identified as the end-to-end dominant resource for each NSI. According to the multi-domain DRF mechanism \cite{halabian2019distributed}, one may equalize $z_n$ for the three NSIs, which requires $x_2=x_3=x_1/2$. It can be observed that NSI traffic volumes can be increased up $x_2=x_3=x_1/2=4$, before the CPU turns to a bottleneck (i.e., fully booked) at the CRAN. Evidently, the second and third NSIs do not obtain a \emph{fair share} of the bottleneck resource (i.e., CPU of the CRAN) under the multi-domain DRF allocation in this example. Moreover, the provisioned capacity to each of them is less than that achievable under a uniform allocation (where $x_n=20/3$, for $n=1, 2, 3$).

\begin{figure}[h!]
\centering
\includegraphics[width = 0.96\columnwidth]{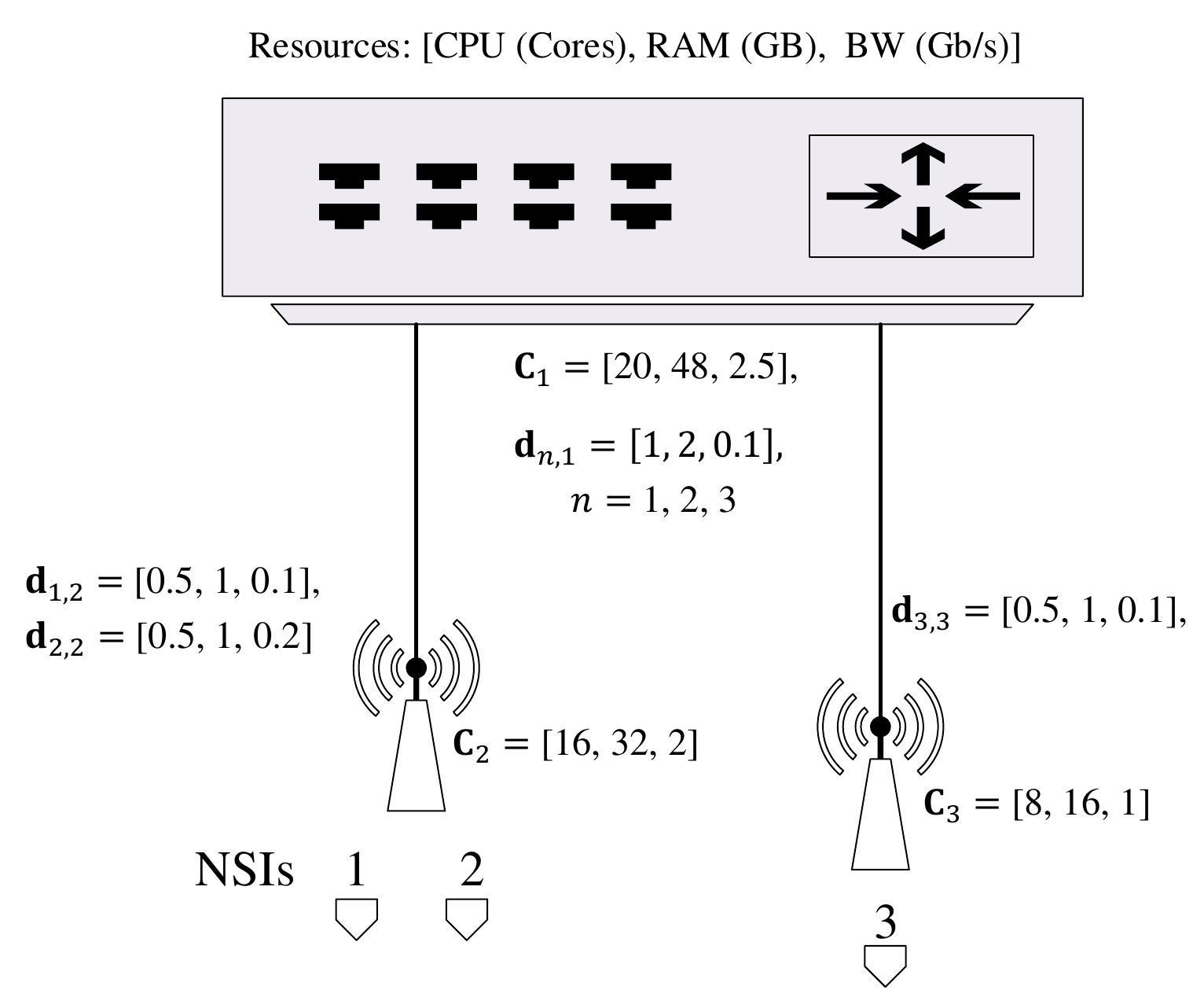}
\caption{\footnotesize A sample network of 2 APs, and 1 CRAN unit, providing resources to three NSIs.}
\label{fig:counterExmpl}
\end{figure}

To address this issue, we propose an extension to DRF which is inspired by the
 per-server dominant share fair (PS-DSF) allocation mechanism presented in \cite{khamse2018, khamse2017per}
 for fair resource allocation from a set of multi-resource heterogeneous servers.
 %Particularly, for each NSI using a particular routing path one may identify a {per-domain dominant resource}
 Particularly, by the PS-DSF mechanism a \emph{per-server dominant resource} is identified for each network function with respect to each server.
 Then, each server strives to maximize the minimum \emph{per-server dominant share} among different network functions~\cite{khamse2017per}.
 For fair resource provisioning in a multi-domain mobile network environment,
 we may identify a \emph{per-domain dominant resource} for each NSI with respect to each node over an end-to-end routing path.
  Then starting from the last domain (i.e., the core network), at each node
 one may allocate a fair-share of the per-domain dominant resource to all NSIs  which are passing through the same node.
 The same procedure can be implemented at preceding domains, except that the allocated traffic volume to each NSI
 is limited by the one at subsequent domains. We refer to this mechanism as \emph{Per-Domain DRF}\footnote{It can be shown that the per-domain DRF mechanism inherits all of the properties which are satisfied by PS-DSD mechanism~\cite{khamse2018}.
The details of such analysis, however, are out of the scope of this paper.}.
 For example, in Fig.~\ref{fig:counterExmpl}, CPU (bandwidth) is identified as the dominant resource for each NSI at the CRAN unit (each of the APs).
 Starting from the CRAN, each of the NSIs receives a fair share of the per-domain dominant resource (i.e., CPU),
  which results in $x_n=20/3$, $n=1,2,3$. Here the traffic volume for NSIs at each of the APs is limited by that at the CRAN (i.e., the end-to-end bottleneck), so the end-to-end allocation for each NSI is given by $x_n=20/3$, $n=1,2,3$.

  In this particular example, the proportional fairness metric of Problem~3 results in the same allocation (as the per-domain DRF mechanism)
   when setting OPEX to zero and assuming the same weights (i.e., payments) for all NSIs.
  By implementing the DRP mechanism, however,
  each of the NSIs would make a (different) payment which maximizes its payoff function.
 In Section~\ref{sec:Num_res} we compare the performance of the proposed DRP mechanism
  (or equivalently the weighted proportional fairness metric of Problem~3)
 against the \emph{multi-domain DRF}, and
 \emph{per-domain DRF}, respectively, while setting the weight for each NSI as its payment under the DRP mechanism.

\section{Performance Evaluation}\label{sec:Num_res}

\subsection{Simulation Setup}
In this section we evaluate the performance of the DRP mechanism by implementing the proposed algorithm in MATLAB,
and comparing its performance against some existing solutions in the literature.
For the sake of numerical evaluations, we consider a multi-domain network, as shown in Fig.~\ref{fig:systemModel},
comprising ten distributed RAN units (e.g., APs), two CRAN units, and a core network data center.
In the RAN segment, we consider five areas, wherein 2 different APs are assumed in each area.
Particularly, it is assumed that APs are of \emph{2 different types}, each providing four types of resources
that are CPU, RAM, memory bandwidth and communication bandwidth.
The resource capacities for the two types of APs, and for the CRAN and CN data centers are given in Table~\ref{table2}.
The resource capacities for each CRAN data center (and CN data center, respectively) are equivalent to 3 instances of Amazon EC2 C5.4
(one instance of Amazon EC2 C5n.18).
The operational costs are taken according to a uniform distribution in the range of $[0.5,1]\cent $ for unit of RAM and
memory bandwidth, in the range of $[1, 2]\cent $ for 1 core of CPU, and in the range of $[1, 10]\cent$ for 1Gb/s of communication bandwidth
 allocated over unit of time.
It is assumed that APs in areas 2, 3, and 4 have connections to both CRAN units,
while the APs in area 1 (area 5, respectively) are connected only to the first CRAN (second CRAN) unit.
The communication bandwidth at each AP
represents the maximum achievable data rate when using the best modulation and coding scheme (i.e, the \emph{nominal capacity}).
Depending on the performance requirements and or (mid-term) feedbacks from the users in a certain area, however, some NSIs may require lower order modulation and coding schemes, or particular beam-forming and or MIMO transmission schemes which possibly result in a lower (average) achievable data rate. Such complexities can be flexibly reflected in the demand vector for different NSIs.
The demand vector for each NSI is assumed to be fixed in one instantiation of the game. However, it may vary in different instantiations to capture (mid-term) fluctuations in the RAN environment.

Since network slicing is rather a new concept, there are no real-world traces publicly available to use.
So, as in \cite{halabian2019distributed}, the demand vector (for 100 Mb/s of traffic) is generated independently for each NSI $n$, with the values taken according to \emph{independent uniform distributions} in the range of $\text{coeff}_1\times[0.4, 0.8]$ cores for CPU, in the range of $[1, 2]$ GB for RAM, in the range of $\text{coeff}_2\times[50, 100]$ Mb/s for communication bandwidth, and in the range of $[10, 20]$ Mb/s for the memory bandwidth. Independency of the demand vector for various NSIs and across different elements ensures heterogeneity in the generated data-set which is the most important requirement to synthesize a multi-resource data-set~\cite{DRF, khamse2018}.
The parameter $\text{coeff}_1$ is taken to be 2 for demands at the RAN (representing more intensive processing at the RAN), and 1 otherwise.
The parameter $\text{coeff}_2$ is taken randomly from the set $\{1,2,4,6\}$ at the RAN for each NSI, and is chosen to be 2 at other network segments.
The utility function for each NSI $n$, $U_{n,l}(z)$, is chosen from the class of concave utility functions with the derivative (i.e., \emph{marginal benefit}) described as $U'_{n,l}(z)=(\phi_{n,l}/z)^{\alpha_n}$ ~\cite{BG92, kelly1998rate}.
The parameter $\alpha_n$ determines the shape of the utility function (or the \emph{marginal benefit}) for each NSI.
The parameter $\phi_{n,l}$ may present the traffic load for each NSI $n$.
Unless otherwise stated, the parameter $\alpha_n$ for each NSI $n$ is chosen according to a uniform distribution in the range of $[1,2]$.
We study the performance of the proposed mechanism under different loading conditions (as described in Section~\ref{sec:Num_res:Res}).
The budget for each NSI is chosen to be $\$100$ in each area.

\begin{table}
\vspace{+2mm}
\footnotesize
\caption{Data-center/AP resource capacities}
\label{table2}
%\centering
\begin{tabular}{c | c c c c}%{p{8.48cm}}
\hline\hline
Node Type & CPU & RAM & Memory BW & Comm BW\\
~ & ($\#$ cores) & (GBytes) & (Gb/s) & (Gb/s) \\
\hline\hline
AP Type 1 & 16 & 32 & 10 & 1 \\\hline
AP Type 2 & 8 & 16 & 5 & 1 \\\hline
CRAN & 48 & 384 & 4$\times$10 & 7 \\\hline
CN & 96 & 384 & 2$\times$50 & 14 \\\hline
\end{tabular}
\end{table}

\subsection{Simulation Results}\label{sec:Num_res:Res}
First we study the convergence behavior of the proposed {distributed resource provisioning} mechanism.
Next, we compare the performance of the DRP mechanism (in terms of the provisioned capacity, resource utilization
and other performance metrics such as OPEX) against
some heuristic and/or recently developed work in the literature, which show the enhanced performance of the proposed solution.

Fig.~\ref{fig_num_convg} shows the number of iterations that are required for the DRP mechanism
to converge to an $\epsilon$-boundary of the optimal solution to Problem~1 under different loading conditions.
In this experiment, the parameter $\alpha_n$ for each NSI is taken according to a uniform distribution in the range of $[1,2]$.
Unless otherwise stated, we consider a fixed number of $N=50$ NSIs.
The parameter $\phi_{n,l}=\phi$ (representing the traffic demand) for different NSIs is chosen under the high-loading condition such that a fully booked resource (i.e., \emph{bottleneck}) exists on each routing path. The traffic demand is reduced to half (and one fourth, respectively) for mid-loading (low-loading) condition. We also consider another high-loading condition where we assume $2N=100$ number of NSIs but $\phi_n$ is halved.
It is observed that a precision of $10^{-3}$ is achieved in the worst case (i.e., under a high loading condition)
within only a few hundreds of iterations, which only takes a few milliseconds when implemented in such a cluster with tens of servers.
Our observations indicate that the convergence rate mainly depends on the overall traffic demand.
Hence, the convergence rate for a high loading regime remains the same when the number of NSIs is doubled but $\phi_n$ is halved (see Fig.~\ref{fig_num_convg}).

\begin{figure}[h!]
\centering
\includegraphics[width = 0.96\columnwidth]{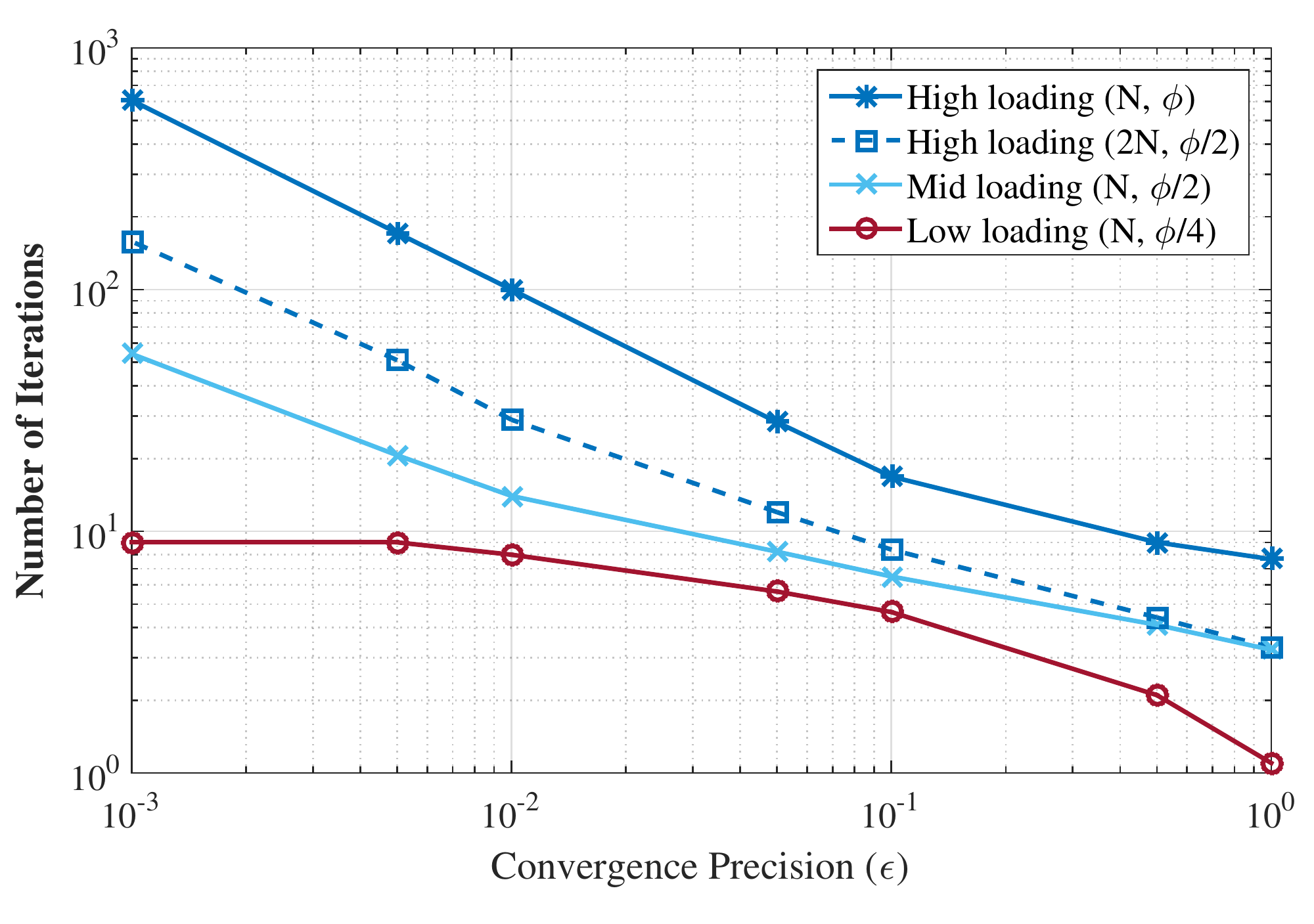}
%\footnotesize
\caption{The required number of iterations for convergence of the DRP mechanism to an $\epsilon$-boundary
of the optimal solution to Problem~1 under different loading conditions. The parameter $\alpha_n$ for each NSI is taken according to a uniform distribution in the range of $[1,2]$.}
\label{fig_num_convg}
\end{figure}

It is worth noting that the proposed mechanism converges within only a few (less than 10) iterations under a
low-loading condition. Intuitively, there is less contention for different resources under lower loading conditions,
resulting in partially booked resources with a fixed pricing at almost all nodes.
With less variations in pricing, the proposed auction game converges more rapidly to the optimal solution under lower loading conditions.

In another experiment reported in Fig.~\ref{fig_num_convg2}, we study the convergence performance of the DRP mechanism under high loading conditions, while the parameter $\alpha_n$ for each NSI is taken in two different ranges. It is observed that the convergence facilitates when $\alpha_n$ takes on values in a tighter range. Intuitively, when $\alpha_n$ takes on values in the range of $[1,1.5]$ (compared to taking values in $[1,2]$),
the shape of utility function and also the marginal benefit for different NSIs get more similar and closer to each other.
It means that the payments from different NSIs will be in a tighter range. Also, the NSIs would be making a smaller change in their payments in response to a change in the resource prices. This implies that a stable price (or the NE) can be established within a fewer number of iterations.

\begin{figure}[h!]
\centering
\includegraphics[width = 0.96\columnwidth]{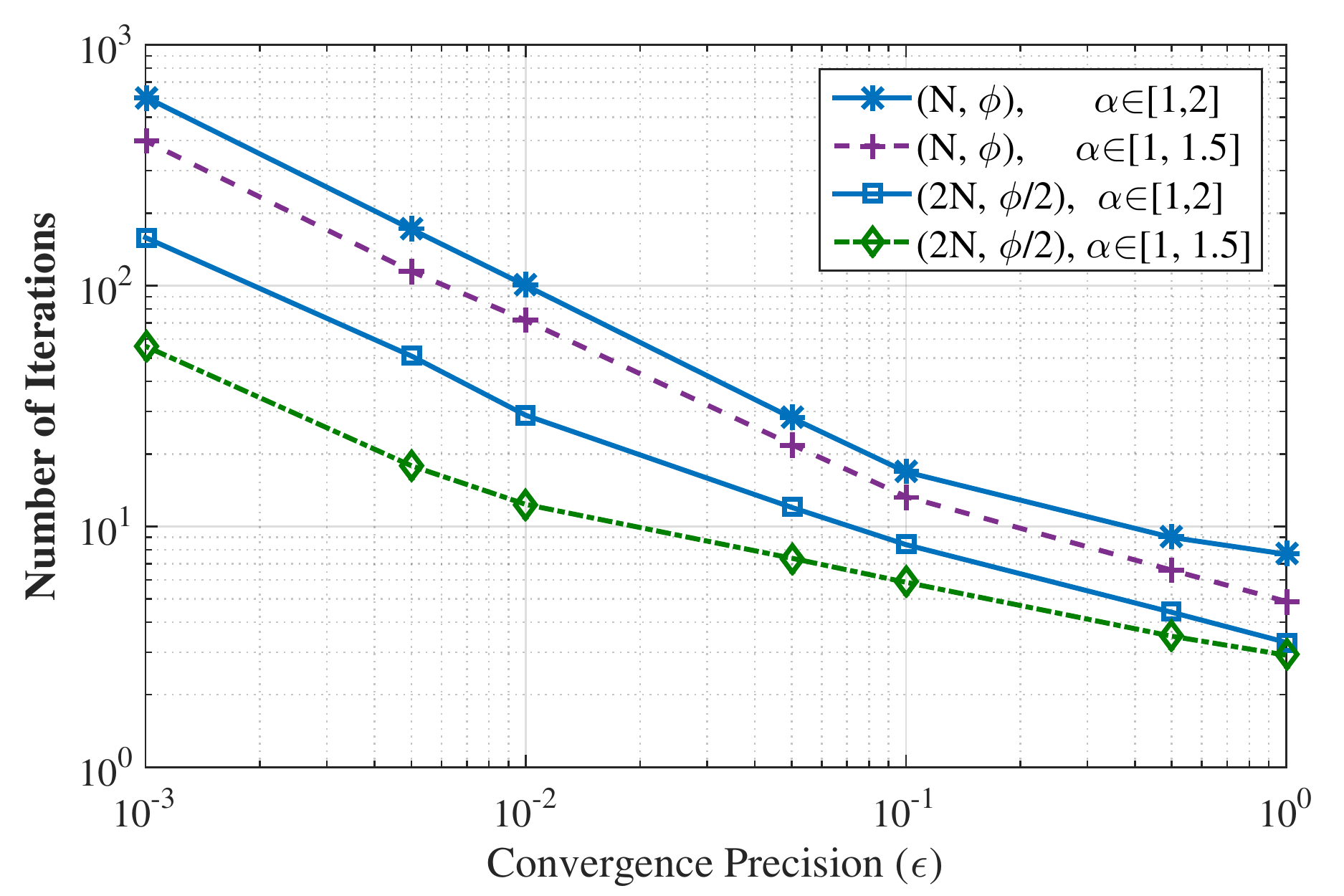}
%\footnotesize
\caption{The required number of iterations for convergence of the DRP mechanism to an $\epsilon$-boundary
of the optimal solution to Problem~1 under high loading conditions. The convergence performance is compared when $\alpha_n$ is taken (according to a uniform distribution) in two different ranges.}
\label{fig_num_convg2}
\end{figure}

To evaluate the performance of the proposed DRP mechanism in terms of \emph{resource utilization},
we compare it against some heuristic and \emph{greedy solutions}
which strive to allocate the whole resources according to some fairness criteria.
In particular, we compare the DRP mechanism with the multi-domain DRF mechanism which allocates resources to different NSIs
by employing dominant resource fairness (DRF \cite{DRF}) across different domains~\cite{halabian2019distributed} (c.f. Section~\ref{sec::ext::MRF}). We show that, however, the \emph{multi-domain DRF} mechanism does not satisfy the sharing incentive property. So, we also implement an extension to DRF, referred to as \emph{per-domain DRF} (c.f. Section~\ref{sec::ext::MRF}), which is shown to satisfy the sharing-incentive property.
We further compare the performance of these three mechanisms (i.e., DRP, multi-domain DRF, and per-domain DRF)
with a generic uniform allocation (described in Section~\ref{sec:main:ChSol}).

To observe how each of the above-described mechanisms performs compared to the uniform allocation, we find
  the allocated traffic volume to each NSI (under each allocation mechanism), and \emph{normalize} it by the allocated traffic under the uniform allocation. Such a parameter, denoted by $r_n$ for each NSI $n$, represents the improvement ratio by which the allocated traffic to NSI $n$ is increased compared to the uniform allocation. In Fig.~\ref{fig:numRes1} we plot $r_n$ for different NSIs, when implementing each of these mechanism for the same NSIs under a high loading condition. It can be observed that with multi-domain DRF, the allocated traffic volume for
two NSIs (index 21 and 24) is less than their allocated traffic volume under the uniform allocation.
It means that the multi-domain DRF does not satisfy the sharing incentive property.
However, it is observed that both of the DRP and per-domain DRF mechanisms out-perform the uniform allocation.
To affirm this observation, we repeat the same experiment for 100 times, generating demand profiles randomly each time.
Then, we find the empirical probability that the improvement ratio for an arbitrary NSI is greater than certain values, that is $P(R>r)$.
Fig.~\ref{fig:numRes11} shows that both of the DRP and per-domain DRF mechanisms always outperform the uniform allocation.
Moreover, the DRP mechanism is shown to provision (on average) more traffic volumes to different NSIs.

\begin{figure}[h!]
\centering
\includegraphics[width = 0.99\columnwidth]{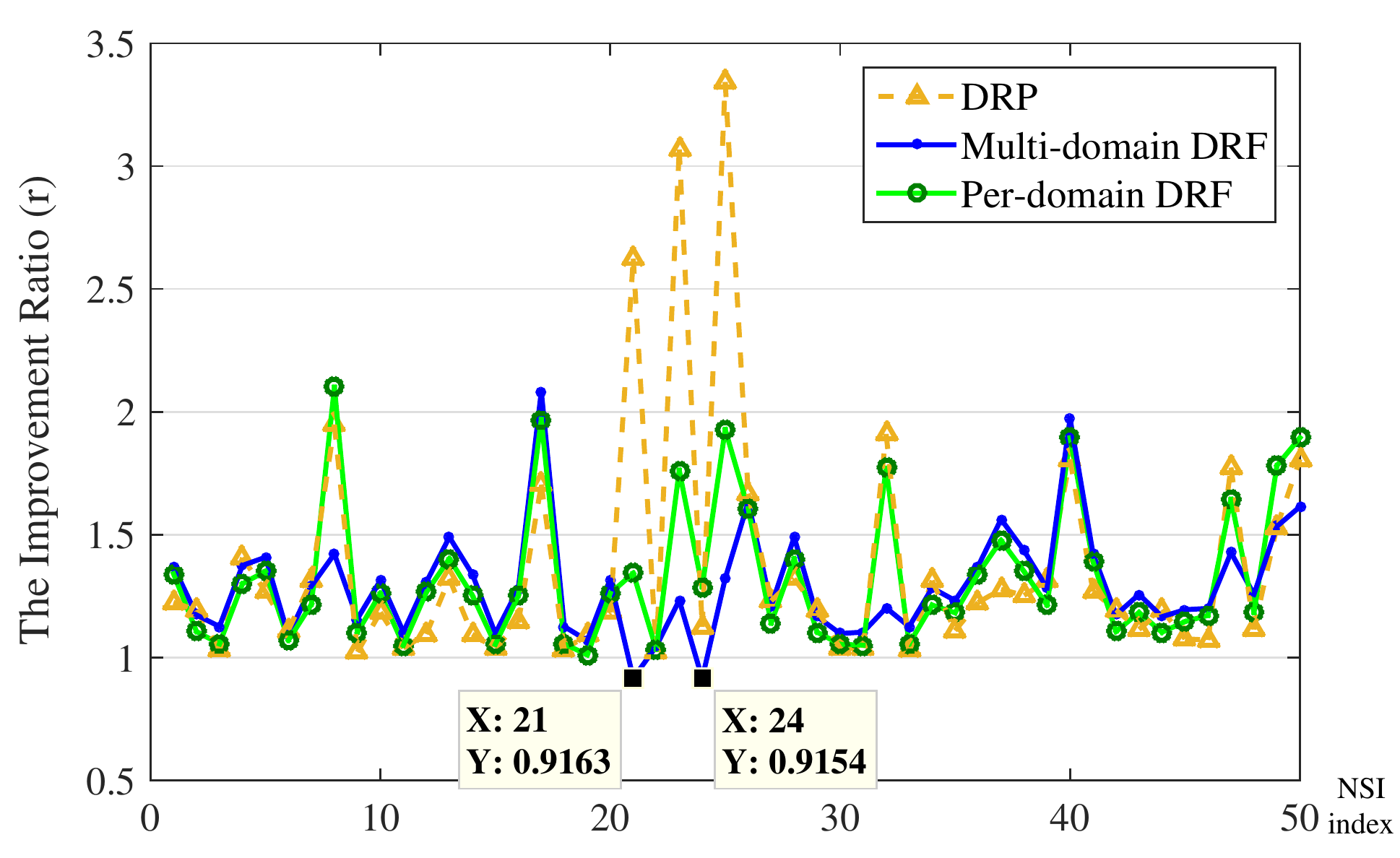}
%\footnotesize
\caption{Comparing the allocated traffic volume to each NSI under different resource provisioning mechanisms when normalized by the allocated traffic volume under the uniform allocation.}
\label{fig:numRes1}
\end{figure}

\begin{figure}[h!]
\centering
\includegraphics[width = 0.99\columnwidth]{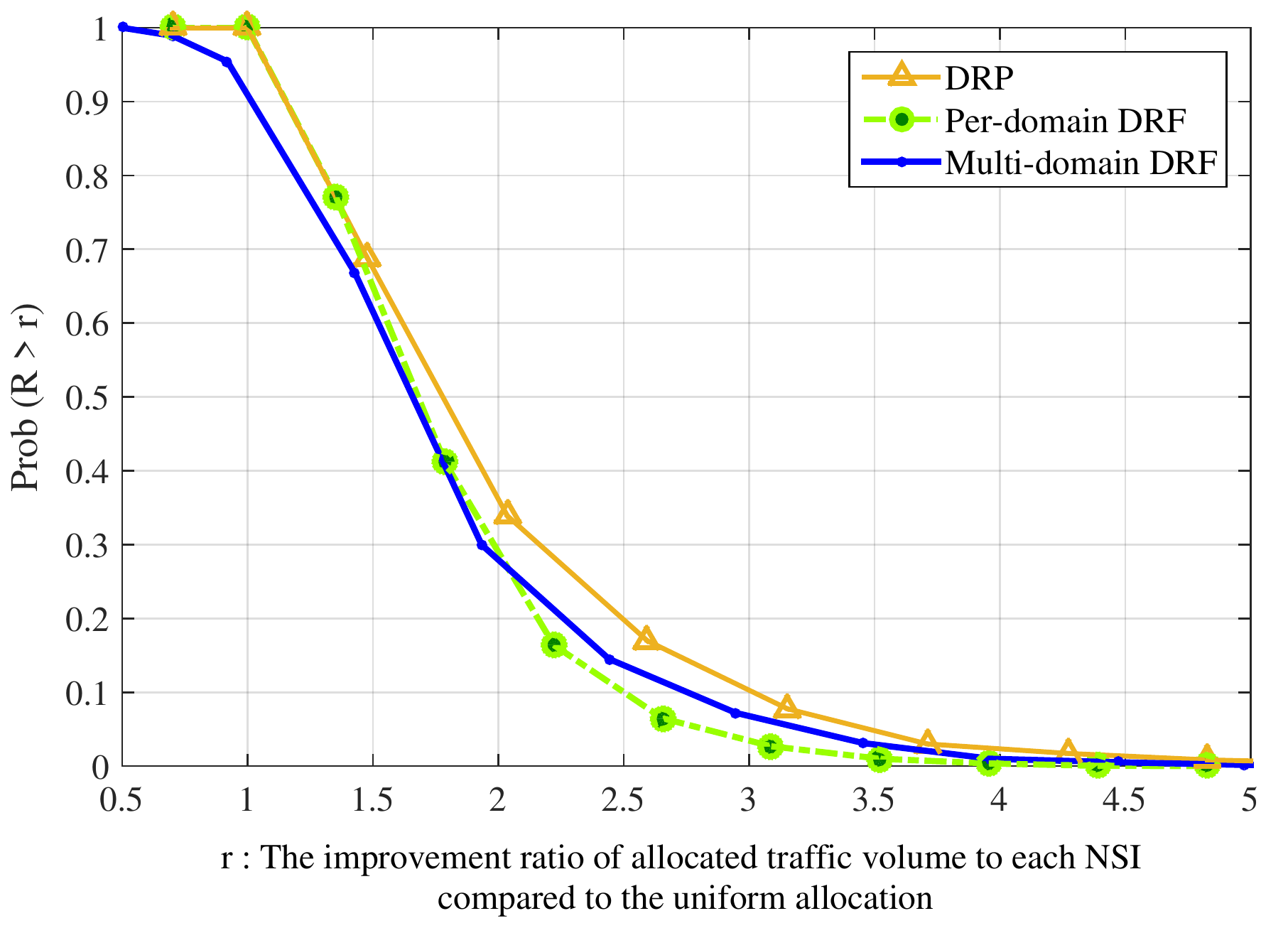}
%\footnotesize
\caption{The empirical probability that the normalized traffic volume for an arbitrary NSI is greater than certain values.}
\label{fig:numRes11}
\end{figure}

\begin{figure}[h!]
\centering
\includegraphics[width = 0.75\columnwidth]{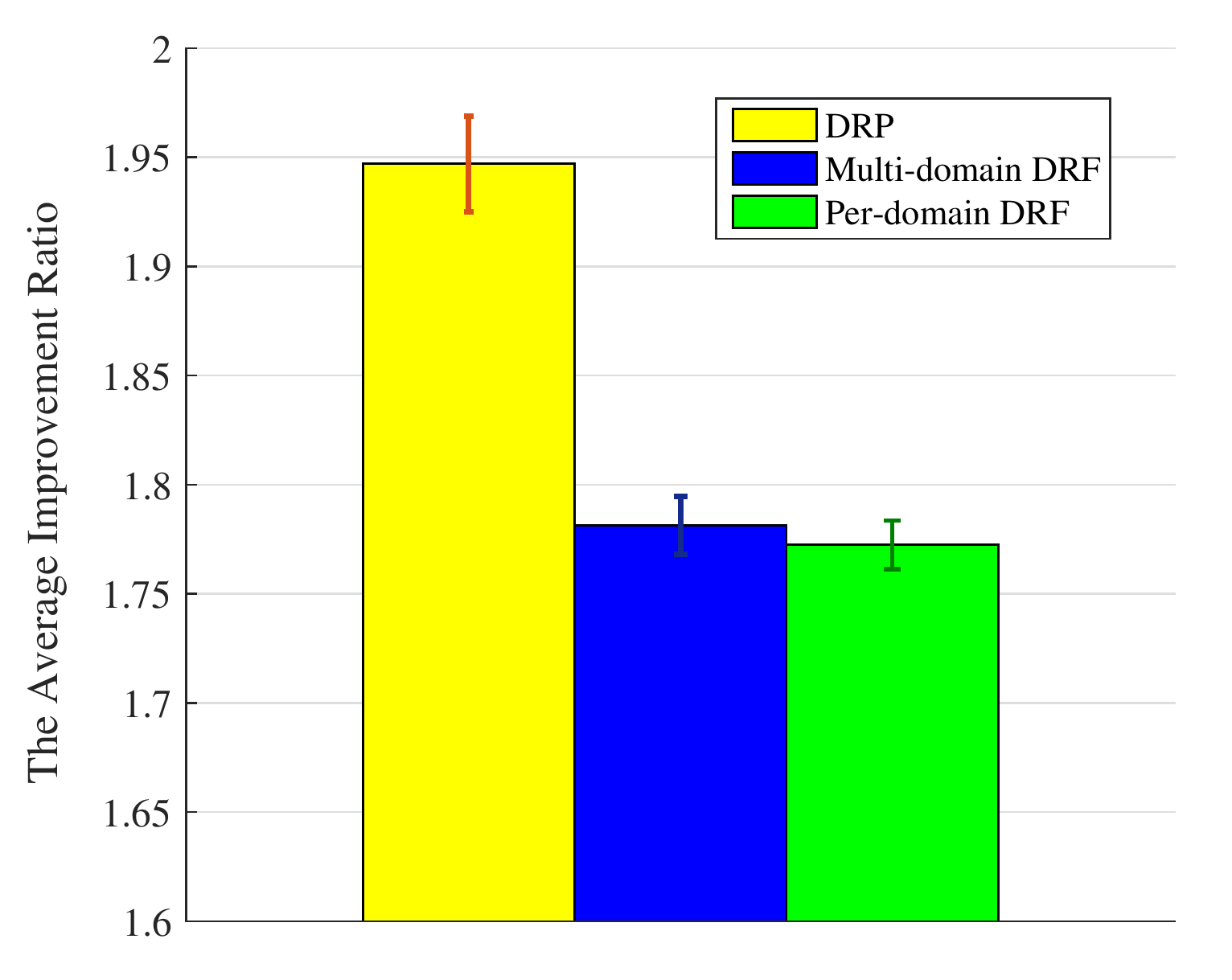}
%\footnotesize
\caption{The average improvement ratio of the end-to-end  capacity provisioned to each NSI compared to the uniform allocation.}
\label{fig:numRes_E2ECAP}
\end{figure}

Particularly, as plotted in Fig.~\ref{fig:numRes_E2ECAP}, the end-to-end capacity provisioned to each NSI (under a high-loading condition) is enhanced under the DRP mechanism by  $95\%$ compared to the uniform allocation, and by at least $10\%$ compared to the two other schemes. The end-to-end queuing delay for each NSI, on the other hand, depends on the provisioned capacity as well as the average traffic load that is admitted to the network. Specifically, if $\rho_{n,l}$ denotes the average traffic load that is admitted for an NSI in a particular area, the end-to-end queuing delay is conversely proportional to $x_{n,l}-\rho_{n,l}$~\cite{ETEDelay5G}. Indeed, a detailed quantitative analysis of delay depends on the admission and flow control policy, as well as other packet-level networking modules (including traffic shaping, packet segmentation, aggregation, or duplication  functionalities), which are out of the scope of this paper. Nevertheless, we can perform a comparative analysis for the end-to-end delay, comparing every two mechanisms while considering a certain traffic load for each NSI. Specifically, assume that the traffic load for each NSI is set to $95\%$ of the (minimum) capacity that can be provisioned under the DRP and every other resource provisioning mechanism. The improvement ratio of the end-to-end delay for the DRP mechanism compared to other schemes is summarized in Table~\ref{table4}. Indeed, an improvement of 10 to 15\% in the capacity provisioned by the DRP mechanism results in reducing the end-to-end queuing delay by a factor of 3 to 4 (compared to the per-domain or multi-domain DRF) under a high-loading condition.

\begin{table}
\vspace{+2mm}
\footnotesize
\caption{The average improvement ratio for the end-to-end delay of the DRP mechanism compared to other schemes.}
\label{table4}
\centering
\begin{tabular}{| l | c | c | c |}%{p{8.48cm}}
\hline
\footnotesize\bf
Delay Improvement Ratio & Uniform & MD-DRF  & PD-DRF\\
\hline\hline
DRP over other schemes  & 19.94   &  4.08       &  3.24 \\\hline
\end{tabular}
\end{table}

To better observe the efficiency of the DRP mechanism in utilizing different resources,
the resource utilization that is achieved on average across all nodes under different mechanisms in a \emph{high loading condition}
is shown in Fig.~\ref{fig:numRes20}. The average resource utilization over different nodes in each domain is also shown in Fig.~\ref{fig:numRes21}-\ref{fig:numRes23}, respectively. All of the results are averaged over 100 experiments.
 The 95\% confidence interval is shown on the top of each bar graph.
 Our observations indicate that at least one of the resources (either CPU or communication bandwidth) is fully booked over each node in Domain 1 under each of the DRP, multi-domain DRF, and per-domain DRF mechanisms.
It means that there is a bottleneck (imposed by the limited resource capacities) on each routing path under a high loading condition.
Despite the greedy nature of the per-domain DRF and multi-domain DRF mechanisms,
it is observed that the average resource utilization that is achieved by the DRP mechanism is increased by around 10\% for all of the resources compared to the per-domain DRF and multi-domain DRF mechanisms, and by 20\% compared to the uniform allocation (see Fig.~\ref{fig:numRes20}).

\begin{figure}[h!]
\centering
\includegraphics[width = 0.99\columnwidth]{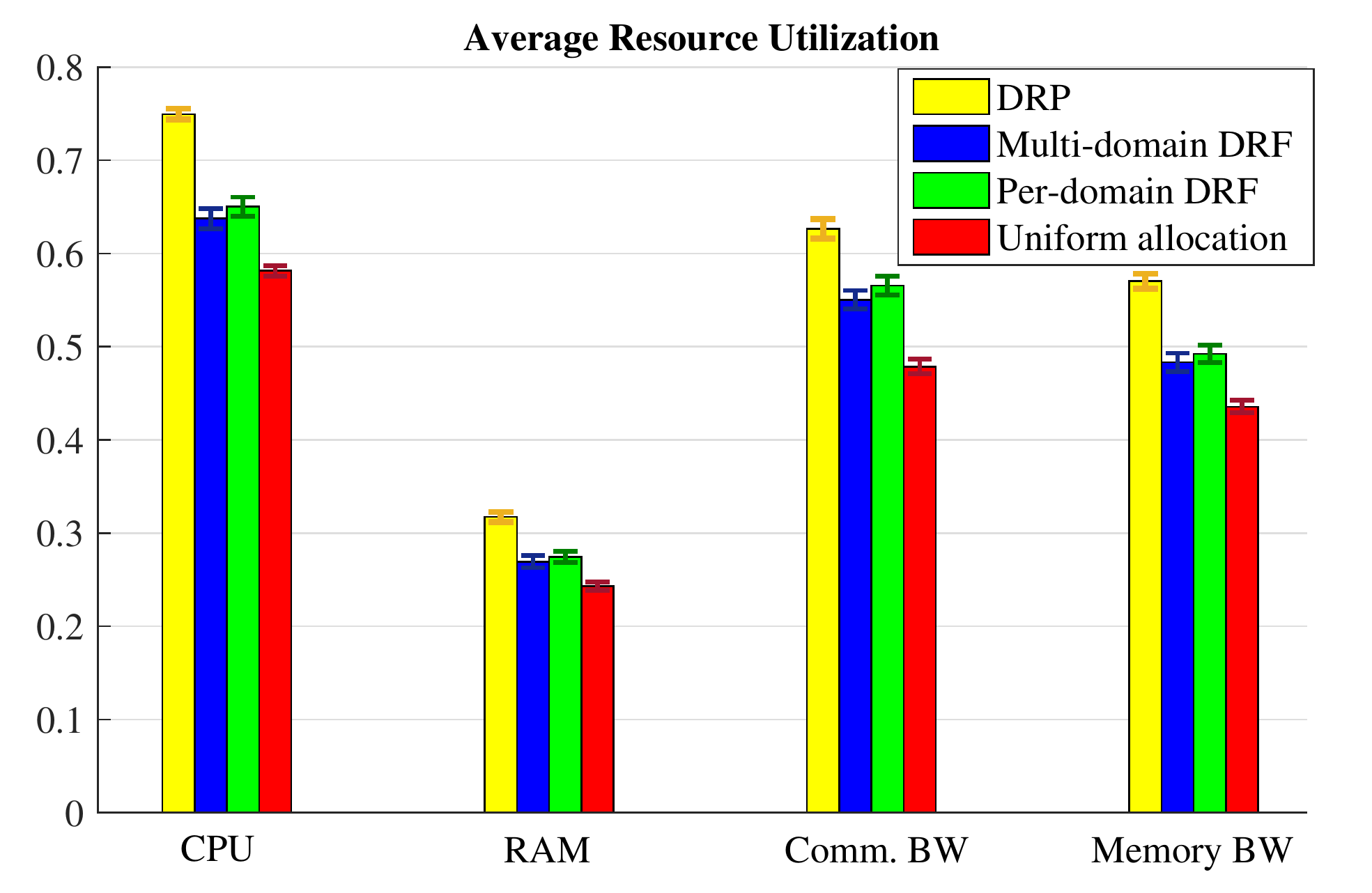}
%\footnotesize
\caption{The resource utilization when averaged over different nodes and over 100 experiments for different allocation mechanisms (from left to the right: DRP, multi-domain DRF, per-domain DRF, and uniform allocation) in a high-loading condition.}
\label{fig:numRes20}
\end{figure}

\begin{figure}[h!]
\centering
\includegraphics[width = 0.99\columnwidth]{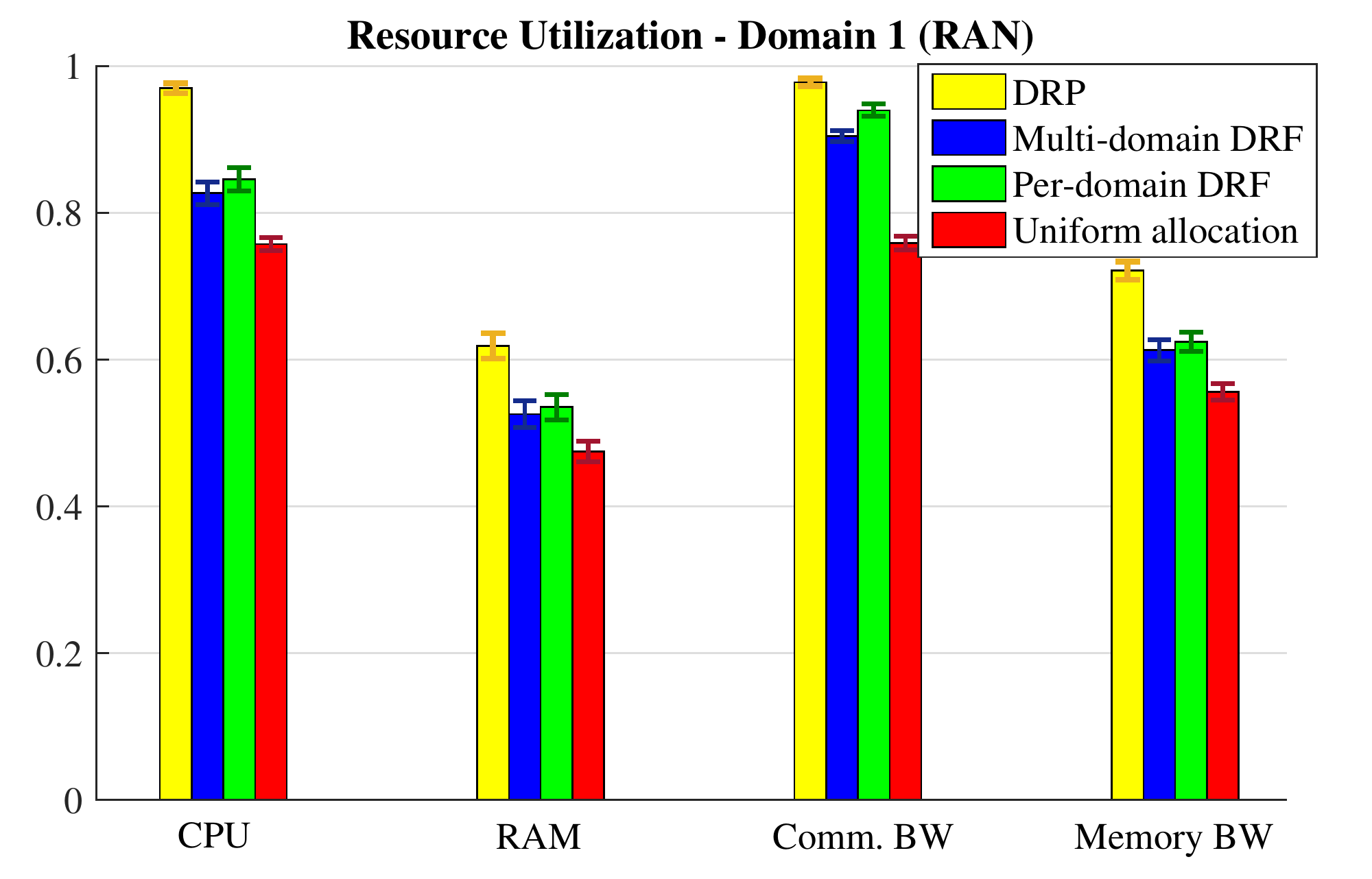}
%\footnotesize
\caption{The resource utilization that is achieved on average over different nodes in Domain 1 (i.e., RAN) under different allocation mechanisms.}
\label{fig:numRes21}
\end{figure}

\begin{figure}[h!]
\centering
\includegraphics[width = 0.99\columnwidth]{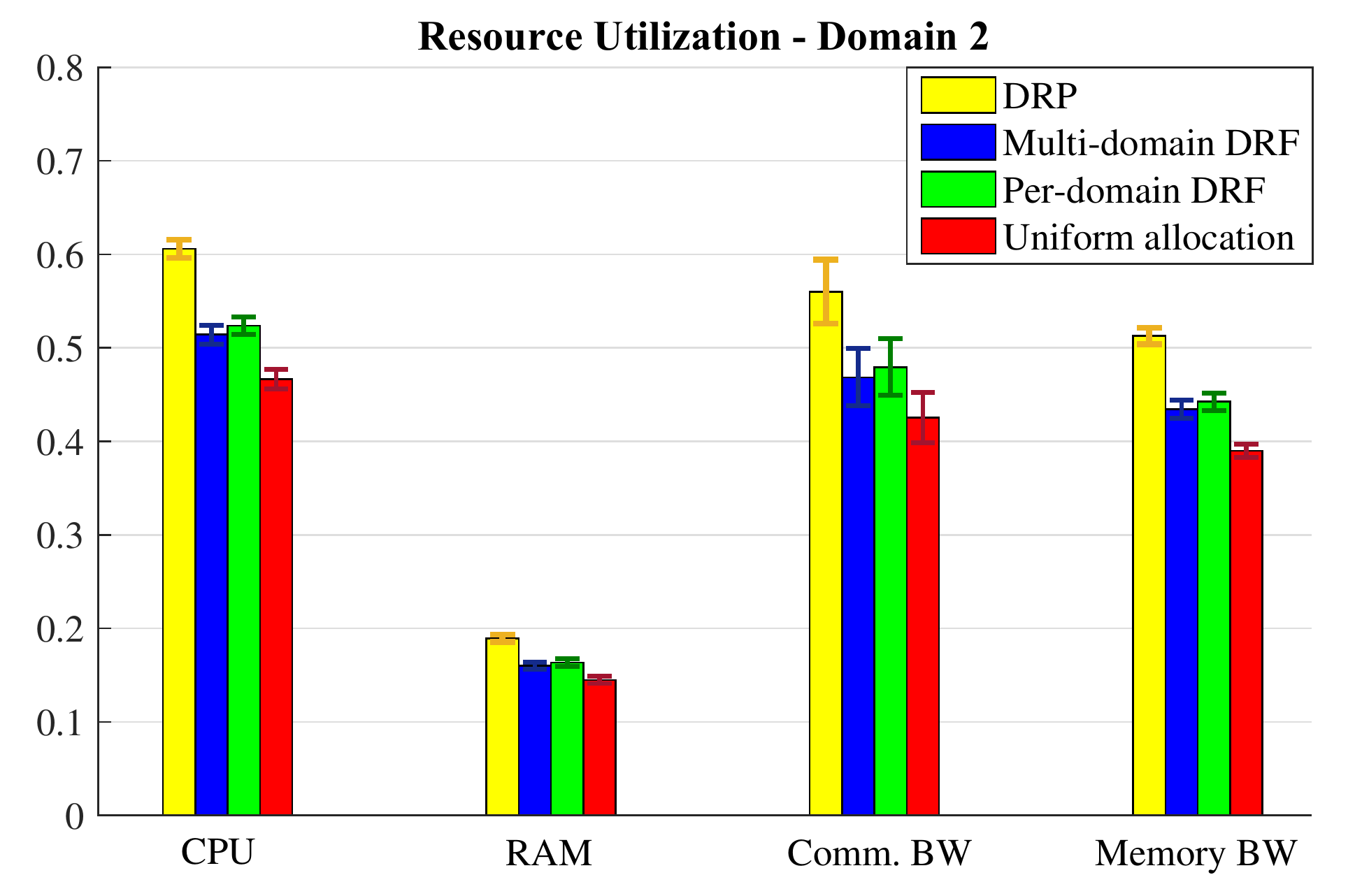}
%\footnotesize
\caption{The resource utilization that is achieved on average over Domain 2 data centers under different allocation mechanisms.}
\label{fig:numRes22}
\end{figure}

\begin{figure}[h!]
\centering
\includegraphics[width = 0.99\columnwidth]{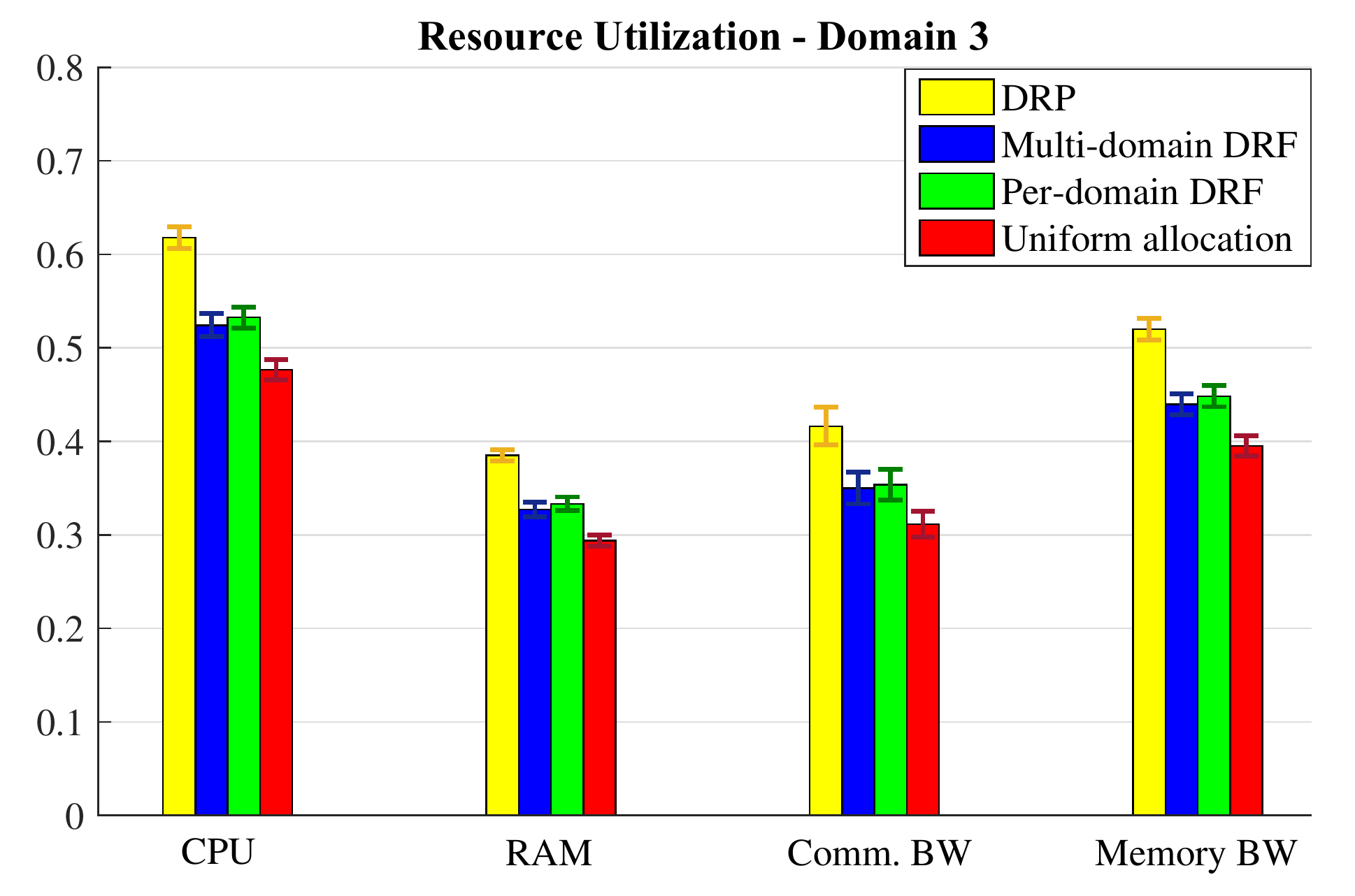}
%\footnotesize
\caption{The average resource utilization for the core network data center under different allocation mechanisms.}
\label{fig:numRes23}
\end{figure}

Another observation that we make here is on the operational expenditures that are imposed (on average) to the network
for provisioning each unit of traffic. The DRP mechanism may limit the allocated resources from different nodes in a low-loading condition
in a way that the operational expenditures are covered by the payment from different NSIs.
Hence, it may not be \emph{fair} to compare the \emph{absolute value} of OPEX (which is considerably reduced by the DRP in a low loading condition),  against greedy mechanisms such as multi-domain DRF or per-domain DRF. To make a fair comparison, we find the average operational expenditures for
\emph{each unit of traffic} under different mechanisms in low-loading and high loading conditions. As shown in Fig.~\ref{fig:numRes3}, the per unit OPEX  is reduced by the DRP mechanism under both low-loading and high-loading conditions.
While the DRP mechanism makes a better utilization of different resources in a high-loading condition,
 yet it results in less OPEX for each unit of traffic, owing to the optimal routing decisions.
The DRP mechanism results in a more considerable reduction in per unit OPEX in a low loading condition.
Intuitively, the DRP mechanism may throttle the allocated resources of costly nodes in a low-loading condition,
while allocating more resources from nodes with low operational costs. Making jointly optimal routing and resource provisioning decisions, the DRP mechanism reduces the per unit OPEX by 12\% compared to the per-domain DRF and multi-domain DRF in a low loading condition (see Fig.~\ref{fig:numRes3}).
It is worth noting that OPEX for other mechanisms does not change much with respect to loading conditions.

\begin{figure}[h!]
\centering
\includegraphics[width = 0.95\columnwidth]{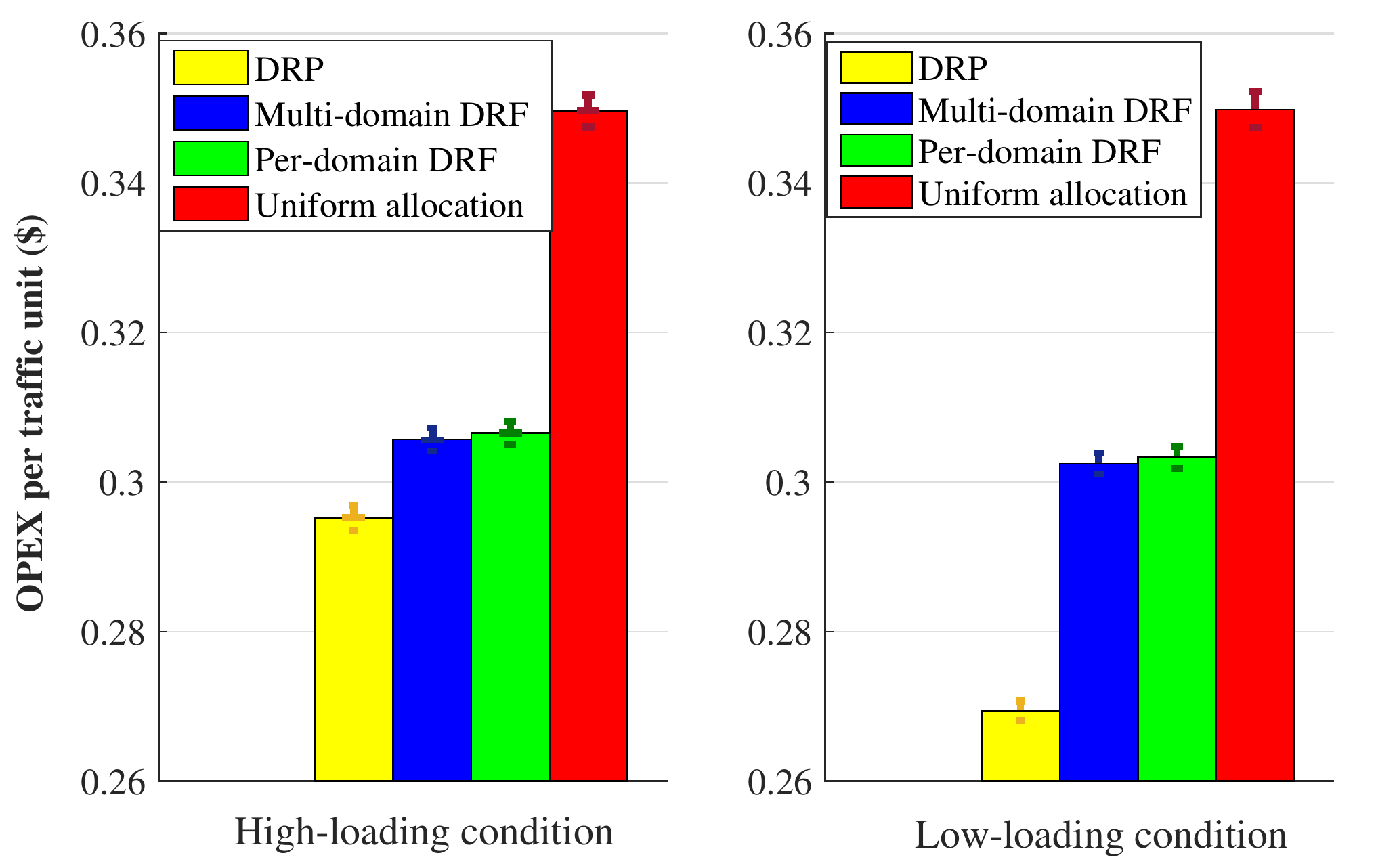}
%\footnotesize
\caption{The average operational expenditures to provision one unit of traffic under different allocation mechanisms.}
\label{fig:numRes3}
\end{figure}

\section{Conclusion}\label{sec:conclusion}
We proposed an agile and distributed mechanism for end-to-end resource provisioning to NSIs in a multi-domain mobile network environment.
In the proposed solution, each network slice tenant finds the optimal traffic volumes for different paths (comprising different chains-of-VNFs)
%volume of traffic which should be routed over each possible path/chain-of-VNFs,
so as to maximize a local payoff function.
Based on the solution to the network slice optimization problem, each network slice tenant decides on the amount of resources which should be acquired
from the service provider(s) in each domain, and accordingly bids for the required resources. Given the payments from different NSIs,
each service provider decides on the resource prices, and then allocates resources to different NSIs. We showed that such an auction game has a
unique NE (in terms of NSI traffic volumes), which maximizes social welfare among network slice tenants, while minimizing OPEX for service providers.
Making optimal routing and resource provisioning decisions while employing a cost-aware resource pricing scheme,
the DRP mechanism is shown to reduce the OPEX for provisioning each unit of traffic,
while enhancing the resource utilization of the infrastructure network at the same time.
The proposed DRP mechanism is distinguished from the existing works in the literature, owing to
 generality of the model, agility of the solution, and the possibility for a distributed implementation
 without sharing any private information among different parties.
The DRP mechanism is superior (in terms of resource utilization and OPEX) not only to the existing solutions,
but also their enhanced versions proposed in this study. %With the advent of the real-world implementation of network slicing, one may
An extension wherein the network slice manager integrates the achievable QoS to the NSI's utility function
can be part of the future work. Also a thorough analysis on the convergence behavior of the DRP mechanism for a variety of practical
 utility functions can be addressed in a future study.

\bibliographystyle{IEEEtran}
\bibliography{scheduling}

\appendix
\begin{proof}[Proof of Theorem~\ref{th_nash_opt}]
Problem~1 is a convex optimization problem. An allocation is an optimal solution to this problem if and only if there exists a set of multipliers $\lambda_{i,r}$, and $\nu_{n}^p$ (corresponding to the constraints in \eqref{global_c1} and \eqref{global_c2}, respectively) so that KKT conditions are satisfied\footnote{The notation $x\perp y$ means  $xy=0$.}~\cite{Boyd}
\be
&& \frac{\partial U_{n,l}(x_{n,l})}{\partial x_{n,l}}=\sum_{i\in p}\sum_r(\lambda_{i,r}+q_{i,r})d^p_{n,i,r} - \nu_{n}^p, ~\forall n,p,\label{KKTP2_1}\qquad\\
&& 0 \le C_{i,r} - \sum_n\sum_{p:i\in p}x_n^pd^p_{n,i,r} \perp \lambda_{i,r}\ge 0, ~~\forall i,r,\label{KKTP2_2}\\
&& 0 \le x_n^p \perp \nu_n^p\ge 0, ~~\forall n,p.  \label{KKTP2_3}
\ee

When ${\bf x}$ is an NE for DRP mechanism in conjunction with some resource prices $\{\mu_{i,r}\}$, it will be an optimal solution to Problem~2
for every NSI $n$. %For given resource prices,
Problem~2 is a convex optimization problem in terms of ${\bf x}_n$ for each NSI $n$. It follows that ${\bf x}_n$ is an optimal solution to Problem~2 if (and only if) there exists a set of multipliers $\{v_{n}^p\}$ so that
\be
&& \frac{\partial U_{n,l}(x_{n,l})}{\partial x_{n,l}}=\sum_{i\in p}\sum_r\mu_{i,r}d^p_{n,i,r} - v_{n}^p, ~~\forall p,\label{KKT_Local_1}\\
&& 0 \le x_n^p \perp v_n^p\ge0, ~~\forall p,  \label{KKT_Local_2}
\ee
where $\mu_{i,r}$ is set according to \eqref{res_price}, so $\mu_{i,r}\ge q_{i,r},~\forall i,r$. We show that the conditions in \eqref{KKTP2_1}-\eqref{KKTP2_3} are satisfied when choosing $\lambda_{i,r}:=\mu_{i,r}-q_{i,r}$, and $\nu^p_{n}=v_n^p$.  It means that we may find a set of multipliers in conjunction with each NE of DRP mechanism, so that the KKT conditions in  \eqref{KKTP2_1}-\eqref{KKTP2_3} are satisfied. This, in turn, implies that each NE of DRP mechanism is an optimal solution to Problem~1.  It is straightforward to reach \eqref{KKTP2_1} and \eqref{KKTP2_3} when choosing $\lambda_{i,r}:=\mu_{i,r}-q_{i,r}\ge0$, and $\nu^p_{n}=v_n^p$.  To observe \eqref{KKTP2_2}, one may substitute for $w_{n,i,r}$ from  \eqref{eq_bids} into \eqref{res_price}, which results in an updated resource price,
\be
\hat{\mu}_{i,r} =   \max\{q_{i,r}, \frac{\sum_n\sum_{p:i\in p}x_n^p d^p_{n,i,r}}{C_{i,r}}\mu_{i,r}\}.\label{res_price_updated}
\ee
At the NE we should have $\hat{\mu}_{i,r}=\mu_{i,r},~\forall i,r$. This is established only if  $\sum_n\sum_{p:i\in p}x_n^p d^p_{n,i,r}\le C_{i,r}$ for every node $i$ and resource $r$ with $\mu_{i,r}=q_{i,r}$, and $\sum_n\sum_{p:i\in p}x_n^p d^p_{n,i,r} = C_{i,r}$ for every node $i$ and resource $r$ with $\mu_{i,r}>q_{i,r}$.
This is exactly equivalent to the condition in \eqref{KKTP2_2} when choosing $\lambda_{i,r}=\mu_{i,r}-q_{i,r}$.

Now, consider an allocation ${\bf x}:=\{x_{n}^p\mid n\in\mathcal{N}, p\in\mathcal{P}_l, l=1,2,..,A\}$, along with a set of multipliers which satisfy the conditions in \eqref{KKTP2_1}-\eqref{KKTP2_3}. We show that ${\bf x}$ is an NE for DRP mechanism in conjunction with the resource prices chosen as $\mu_{i,r}:=q_{i,r}+\lambda_{i,r}$. To have an NE, ${\bf x}_n$ should be an optimal solution to Problem~2 for every NSI $n$.
On the other hand, ${\bf x}_n$ is an optimal solution to Problem~2, if it satisfies conditions in \eqref{KKT_Local_1}-\eqref{KKT_Local_2},
which is the case when choosing $\mu_{i,r}:=q_{i,r}+\lambda_{i,r}$, and $v_n^p=\nu_n^p$ (c.f. \eqref{KKTP2_1} and \eqref{KKTP2_3}).
Finally, the condition in \eqref{KKTP2_2} implies that the resource prices remain steady in DRP mechanism (i.e., $\hat{\mu}_{i,r}=\mu_{i,r}$), when choosing $\mu_{i,r}=q_{i,r}+\lambda_{i,r}$ (c.f. \eqref{res_price_updated}).
\end{proof}

\begin{proof}[Proof of Theorem~\ref{TH_prop_fair}]
Problem~3 describes a convex optimization problem. An allocation ${\bf x}$ % ${\bf x}:=\{x_{n}^p\mid n\in\mathcal{N}, p\in\mathcal{P}_l, l=1,2,..,A\}$
is a solution to this problem if and only if there exists a set of multipliers $\lambda_{i,r}$, and $\nu_{n}^p$ (corresponding to the constraints in \eqref{global_SN_c1} and \eqref{global_SN_c2}, respectively) so that KKT conditions are satisfied~\cite{Boyd}
\be
&& \frac{w^*_{n,l}}{x_{n,l}}=\sum_{i\in p}\sum_r(\lambda_{i,r}+q_{i,r})d^p_{n,i,r} - \nu_{n}^p, ~\forall n,p,\label{KKTP3_1}\qquad\\
&& 0 \le C_{i,r} - \sum_n\sum_{p:i\in p}x_n^pd^p_{n,i,r} \perp \lambda_{i,r}\ge 0, ~~\forall i,r,\label{KKTP3_2}\\
&& 0 \le x_n^p \perp \nu_n^p\ge 0, ~~\forall n,p.  \label{KKTP3_3}
\ee

It should be noted that $w^*_{n,l}$ is the payment made by NSI $n$ in area $l$, when DRP mechanism is in an NE equilibrium.
It means that (c.f. \eqref{paymentPerRegion})
\be
{w^*_{n,l}} = {x^*_{n,l}}  \sum_{i\in p}\sum_r(\lambda^*_{i,r}+q_{i,r})d^p_{n,i,r}, ~p\in{\mathcal P}^*_{n,l},\label{Th2_F4}
\ee
where $x^*_{n,l}$ is the solution at the NE, and ${\mathcal P}^*_{n,l}$ is the set of paths with minimum transmission cost for NSI $n$ in area $l$.
According to Theorem~1, ${\bf x}^*$ is a solution to \eqref{KKTP2_1}-\eqref{KKTP2_3} in conjunction with $\{\lambda^*_{i,r}\}$ and $\{\nu_n^{p*}\}$.
It follows from \eqref{Th2_F4} and \eqref{KKTP2_1}-\eqref{KKTP2_3} that ${\bf x}^*$ is also a solution to \eqref{KKTP3_1}-\eqref{KKTP3_3} when choosing $\lambda_{i,r} = \lambda^*_{i,r},~\forall i,r$ and $\nu_n^p=\nu_n^{p*},~\forall n,p$.
In the same way, it can be observed that every solution to Problem~3 (satisfying the conditions in \eqref{KKTP2_1}-\eqref{KKTP2_3}) is a solution to Problem~1.
\end{proof}

\begin{proof}[Proof of Theorem~\ref{TH_properties}]
According to Theorem~\ref{TH_prop_fair}, an allocation is an NE for the DRP mechanism if and only if it is a solution to Problem~3.
Let $\lambda_{i,r}$, and $\nu_{n}^p$, respectively, denote the Lagrange multipliers corresponding to the constraints in \eqref{global_SN_c1} and \eqref{global_SN_c2}. For every path $p\in\mathcal{P}_l$ and NSI $n$ it follows that
\be
\frac{w^*_{n,l}}{x_{n,l}}=\sum_{i\in p}\sum_r(\lambda_{i,r}+q_{i,r})d^p_{n,i,r} - \nu_{n}^p,\label{Pf3_1}
\ee
where $\nu_n^p=0$, when $x_n^p>0$.

To show \emph{envy-freeness}, we show that each NSI $n$ would not prefer the allocated resources to another NSI $m$ over any path $p$, when adjusted according to their payments. The resources allocated from data center $i$ to NSI $m$ for $x_{m}^p$ unit of traffic is given by $[x_m^pd^p_{m,i,r}]$, $r=1,2,...,M$.
%${\bf a}_{m,i}^p:=[x_m^pd_{m,i,r}]$, $r=1,2,...,M$.
The payment of NSI $m$ for this amount of traffic is given by $w_m^{p*}:=x_m^pw^*_{m,l}/x_{m,l}$.
Such resources are preferred by NSI $n$ if
\be
\frac{x_{n}^pd^p_{n,i,r}}{w_n^{p*}} < \frac{x_{m}^pd^p_{m,i,r}}{w_m^{p*}}, ~\forall r,i\in p,\label{Pf3_2}
\ee
or equivalently (by substituting for $w_n^{p*}$ and $w_m^{p*}$),
\be
\frac{x_{n,l}d^p_{n,i,r}}{w^*_{n,l}} < \frac{x_{m,l}d^p_{m,i,r}}{w^*_{m,l}}, ~\forall r,i\in p.\label{Pf3_3}
\ee

Consider some path $p$ for which $x_m^p>0$, so that $\nu_m^p=0$. It follows from \eqref{Pf3_1} that
\be
\frac{w^*_{m,l}}{x_{m,l}}=\sum_{i\in p}\sum_r\mu_{i,r}d^p_{m,i,r},\label{Pf3_4}\\
\frac{w^*_{n,l}}{x_{n,l}}\le\sum_{i\in p}\sum_r\mu_{i,r}d^p_{n,i,r},\label{Pf3_5}
\ee
where $\mu_{i,r}:=\lambda_{i,r}+q_{i,r}$.
Multiplying both sides of \eqref{Pf3_5} by $x_{n,l}/w^*_{n,l}$, and using the inequality in \eqref{Pf3_3}, result in
\be
1 &\le& \sum_{i\in p}\sum_r\mu_{i,r}\frac{x_{n,l}}{w^*_{n,l}}d^p_{n,i,r}\label{Pf3_6}\\
  & < & \sum_{i\in p}\sum_r\mu_{i,r}\frac{x_{m,l}}{w^*_{m,l}}d^p_{m,i,r}=1,\label{Pf3_7}
\ee
which is a contradiction.

Consider an NE resulting from the DRP mechanism. To prove the \emph{sharing incentive} property we need to show that each NSI is provided with more traffic volume (under the NE) compared to the uniform allocation. To characterize the uniform allocation,
let define $\gamma^p_{n,i}$ as the (maximum) volume of traffic which can be processed for NSI $n$ through path $p$ when monopolizing
the \emph{whole resources} allocated from node $i\in p$ under the NE. That is,
\be
\gamma^p_{n,i}:=\min_{r}\frac{\eta_{i,r}C_{i,r}}{d^p_{n,i,r}}, \label{Pf3_8}
\ee
where $\eta_{i,r}$ is the portion of resource $r$ that is utilized at the NE.
The (maximum) volume of traffic which can be processed for NSI $n$ through path $p$ is given by
\be
x^{p,\text{uni}}_{n}:=\min_{i\in p}\frac{w^{p*}_{n,i}}{W^*_i}\gamma^p_{n,i},
\ee
where
\be
W^*_i:=\sum_{m,p':i\in p'}w^{p*}_{n,i}.
\ee
In the following we show that $x_n^p\ge x^{p,\text{uni}}_{n}$ for every path $p\in\mathcal{P}_l$ with $w^{p*}_{n}>0$.
In particular, for every path $p$ with $w^{p*}_{n}>0$,
%originating from node $i_1$,
if we multiply both sides of \eqref{Pf3_1} by
${w^{p*}_{n,i}}\gamma^p_{n,i}/{W^*_i}$, for every node $i\in p$, it follows that
\be
\gamma^p_{n,i}\frac{w^{p*}_{n,i}}{W^*_i}\frac{w^*_{n,l}}{x_{n,l}} \le
 w^{p*}_{n,i} + \gamma^p_{n,i}\frac{w^{p*}_{n,i}}{W^*_i}\sum_{i'\in p, i'\neq i}\sum_r\mu_{i',r}d^p_{n,i',r}\label{Pf3_11},
\ee
where the inequality follows from the fact that $\gamma^p_{n,i}d^p_{n,i,r}\le \eta_{i,r}C_{i,r},~\forall r$ (c.f. \eqref{Pf3_8}), and
$\sum_r\mu_{i,r}\eta_{i,r}C_{i,r}=W^{*}_i$. In the same way (multiplying both sides of \eqref{Pf3_11} by ${w^{p*}_{n,i'}}\gamma^p_{n,i'}/{W^*_{i'}}$), it can be observed that
\be
\min_{i\in p}\{\gamma^p_{n,i}\frac{w^{p*}_{n,i}}{W^*_i}\}\frac{w^*_{n,l}}{x_{n,l}}\le \sum_{i\in p}w^{p*}_{n,i} = w^{p*}_n.
\ee
On the other hand, $\frac{w^*_{n,l}}{x_{n,l}}=\frac{w^{p*}_{n}}{x^p_{n}}$ for every path $p\in\mathcal{P}_l$ with $w^{p*}_{n}>0$.
Hence,
\be
\frac{w^{p*}_{n}}{x^p_{n}} x^{p,\text{uni}}_{n}  \le w^{p*}_n,
\ee
or, $x_n^p\ge x^{p,\text{uni}}_{n}$.
\end{proof}

%\newpage\clearpage

\end{document}